\RequirePackage{amsmath}
\RequirePackage{thm-restate}
\documentclass[runningheads]{llncs}
\usepackage{amssymb,verbatim}
\usepackage{latexsym}
\usepackage[dvipsnames,usenames]{color}
\usepackage[hmargin=1.2in,vmargin=1.2in]{geometry}%\usepackage[pdftex,pagebackref,letterpaper=true,colorlinks=true,pdfpagemode=none,urlcolor=blue,linkcolor=blue,citecolor=BrickRed,pdfstartview=FitH]{hyperref}
\usepackage[pagebackref]{hyperref}
\usepackage{dsfont}

 \providecommand{\F}{\mathbb{F}}

\date{}

\declaretheorem[name=Lemma]{lem}
\declaretheorem[name=Corollary]{cor}
\newtheorem{prop}[lem]{Proposition}

\newtheorem{defn}{Definition}
\newtheorem{rmk}{Remark}

\DeclareMathOperator*{\argmin}{arg\,min}
\def \mC {\mathcal{C}}

\def \mC {\mathcal{C}}

\def \Xi {{X^{[i]}}}

\onecolumn

\title{Leakage-Resilient Secret Sharing with Constant Share Size}
\author{Ivan Tjuawinata\thanks{I. Tjuawinata is with the Strategic Centre for Research on Privacy-Preserving Technologies and Systems, Nanyang Technological University, Singapore 637553 (e-mail: ivan.tjuawinata@ntu.edu.sg).}%
,Chaoping Xing\thanks{C. Xing is with School of Electronic Information and Electric Engineering, Shanghai Jiao Tong University, Shanghai, 200240 China (e-mail: xingcp@sjtu.edu.cn).}%
}%
\institute{}
\begin{document}

\maketitle

\begin{abstract}
In this work, we consider the leakage resilience of algebraic-geometric (AG for short) codes based \textit{ramp} secret sharing schemes extending the analysis on the leakage resilience of linear threshold secret sharing schemes over prime fields that is done by Benhamouda et al. Since there does not exist any explicit efficient construction of AG codes over prime fields, we consider constructions over prime fields with the help of concatenation method and constructions of codes over field extensions. Extending the Fourier analysis done by Benhamouda et al., one can show that concatenated algebraic geometric codes over prime fields do produce some nice leakage-resilient secret sharing schemes.

One natural and curious question is whether AG codes over extension fields produce better leakage-resilient secret sharing schemes than the construction based on concatenated AG codes. Such construction provides another advantage compared to the construction over prime fields using concatenation method. It is clear that AG codes over extension fields give secret sharing schemes with smaller reconstruction for a fixed privacy parameter $t.$ In this work, it is also confirmed that indeed AG codes over extension fields have stronger leakage-resilience under some reasonable assumptions. These advantages strongly motivate the study of secret sharing schemes from AG codes over extension fields.

The current paper has two main contributions: (i) we obtain leakage-resilient secret sharing schemes with constant share sizes and unbounded numbers of players. (ii) via a sophisticated  Fourier Analysis, we analyze the leakage-resilience of secret sharing schemes from codes over extension fields. This is of its own theoretical interest independent of its application to secret sharing schemes from algebraic geometric codes over extension fields.
\end{abstract}

\section{Introduction}
%A secret sharing scheme enables a secret to be distributed among a group of players where each player receives a share. It has shown to be very important in the field of modern cryptography. 
A secret sharing scheme, which enable a secret to be distributed among a group of players where each player receives a share, is a very important building block in the study of modern cryptography.
Intuitively, it allows for any authorized set of players to use their share to recover the original secret while the shares of any forbidden set of players contain no information of the original secret. It has found many applications in fields such as secure multiparty computation~\cite{BGW88,CCD88,SPDZ12}, distributed authorities~\cite{BF01,Ngu05}, fair exchange~\cite{AV04}, electronic voting~\cite{Ngu05,Sch99} and threshold cryptography~\cite{Des87,DF89,Ift06,Sho00,SG02}. There are also different variants of secret sharing schemes that provide different functionalities such as proactive secret sharing scheme~\cite{HJKY95,SW00,ZSV05}, verifiable secret sharing~\cite{Fel87,RB89,Ped92,Sta96} and computationally secure secret sharing scheme~\cite{Kra94,RP11,CLM17}.

All the variants we discussed so far have the same basic assumption, namely, any share is either fully corrupted or completely hidden from the adversary. However, such assumption may not be reasonable for all cases. Side-channel attacks may enable the adversary to get some partial information about all the shares. In this case, the previous investigations do not provide any privacy guarantee. Such leak can also be shown to have large risk towards the privacy of the secret. A simple example can be found from additive secret sharing scheme over the field $\mathbb{F}_{2^w}$ for some positive integer $w.$ In this scenario, it can be shown that leakage of just one bit from each share may leak a bit of the secret which can be used as a distinguisher between shares generated using the additive secret sharing scheme from uniformly random strings. Guruswami and Wootters~\cite{GW17} showed that when Shamir's secret sharing scheme is used in some settings, a full recovery of the whole secret is even possible from just one-bit leakage from each share. Due to the high usage of secret sharing schemes over $\F_{2^w}$ \cite{SPDZ2k,AFLNO16}, such vulnerability is important to be considered. To address this, recently, there has been a new research direction which considers the security of secret sharing scheme under such leakage~\cite{DP07,BGK14,GK18,GK182,ADN19,SV19,BDIR19,NS20}.

The existence of leakage resilient secret sharing schemes have been shown to be related to other fields. A leakage resilient secret sharing scheme can be used to build an MPC scheme that is secure against semi-honest adversary enjoying some local leakage of shares of uncorrupted players~\cite{BDIR19}. In another direction, having such leakage resilient secret sharing scheme shows that we cannot have the same secret sharing scheme or the related code to have the regenerating property with the same bandwidth since it is shown that any leak of such magnitude should not be sufficient to recover much information about the original secret.

\subsection{Existing Results}
Leakage resilient cryptography is a research topic that has attracted many attentions, (see for example \cite{DKL09,Dz06,GO96,ISW03,Koc96,KJJ99}). In particular, there have also been quite extensive studies on secret sharing schemes providing resilience against local leakage. Such research direction was first considered by Dziembowski and Pietrzak~\cite{DP07} and it can be mainly divided to two directions. The first direction is the study of non-linear secret sharing schemes specially constructed for their leakage resilience properties (see for example \cite{GK18,KMS18}). For this work, we are following another direction, which is to study a more general family of linear secret sharing schemes for their leakage resilience properties. We will discuss in more detail the works in this direction.

Recently, Benhamouda et al.~\cite{BDIR19} investigated the effect of leakage of partial information of all shares when threshold secret sharing schemes that are based on linear Maximum Distance Separable (MDS for short) codes over prime fields are considered. More specifically, given a secret element $s$ of a finite field $\F_q,$ the share for each party can be generated from linear combination(s) of $s$ along with some random field elements. Local leakage from each player can then be extracted as a function of his share. Such study was inspired by a result on regenerating codes by Guruswami and Wootters~\cite{GW17}. In their work, they discovered that in some settings, a secret that is being shared using Shamir's secret sharing schemes over any finite field of characteristic $2$ may be completely recovered just by using $1$ bit of leakage from each share. In order to investigate the extent of such attack, Benhamouda et al. considered the leakage resilience of general linear threshold secret sharing schemes over prime fields. The resilience of a scheme to the leakage can be measured by finding the statistical distance between the leak from different possible secrets. Intuitively, a secret sharing scheme is local leakage resilient against $\theta$ corruption and $\mu$-bit leakage if given the full shares of any $\theta$ players along with any $\mu$-bit information from each of the remaining shares, the adversary cannot learn much information regarding the original secret being secretly shared. By analyzing the leakage resilience of linear MDS codes over prime fields and using the close relation between linear threshold secret sharing schemes and linear MDS codes, Benhamouda et al. provided some leakage resilience measure for linear threshold secret sharing schemes over prime fields. Through this analysis, they discovered some families of additive secret sharing schemes that provide leakage resilience even when less than $1$ bit of randomness remains from each share. They have also identified some Shamir's secret sharing schemes that provide leakage resilience when a constant fraction of each share is leaked.

%Such computation for shares defined as entries of codewords belonging to an MDS code can then be done using Fourier Analysis which establishes some relation between the statistical distance of the leak with some summation figure based on the dual of the code.
%
%Based on the underlying Fourier analysis, which is established over prime fields, Benhamouda et al. provided the statistical distance analysis for leakage from entries of codewords from a linear MDS code defined over prime fields. Such analysis is done by establishing the relation between the statistical distance of the leak with the sum of some Fourier coefficients of some specific functions using the Poisson summation formula. Such sum can then be seen as a sum of $s$ distinct $p$-th roots of unity for some positive integer $s$ which can then be upper bounded by the sum of the first $s$ $p$-th roots of unity. Utilizing Cauchy-Schwarz inequality, they proceed to provide an alternative upper bound. The investigation is then applied to additive secret sharing schemes and Shamir's secret sharing schemes to provide their resilience against local leakage. Having such results, they can then provide some classes of additive and Shamir's secret sharing schemes with good leakage resilience. Some examples of such families include additive secret sharing schemes with all-but-one bit of leakage and Shamir's secret sharing schemes with constant-fraction of leakage with their respective values of parameters.

Such study on the leakage resilience of threshold secret sharing scheme over prime fields is then extended by Maji et al.~\cite{MPSW20}. In their work, they provide an improvement on the leakage resilience of threshold secret sharing schemes over prime fields. They then proved that with overwhelming probability, a secret sharing scheme based on a random linear MDS codes over prime fields is leakage resilient. 

Concurrently, Nielsen and Simkin~\cite{NS20} have also considered the leakage resilience of information theoretic threshold secret sharing schemes. Instead of considering the existence of threshold secret sharing schemes with strong leakage resilience capability, they provided a lower bound for the share length to ensure leakage resilience against unconditional adversary. Combined with the results in~\cite{BDIR19} and~\cite{MPSW20}, these works identify the range of parameters of a Shamir's secret sharing scheme that can provide some leakage resilience capability.

\subsection{Our Contribution}

Although Shamir's secret sharing scheme is widely used in various applications, it has some disadvantages. One of the disadvantages of using Shamir's secret sharing scheme is the requirement of the field size to be larger than the number of players. In consequence, the share size increases with the increase of the number of players. The aim of this paper is to construct an explicit family of leakage resilient \textit{ramp} secret sharing schemes with constant share size and unbounded number of players. In order to achieve this, we consider the use of algebraic geometric codes (AG codes for short). It is a well-known fact that codes over prime fields provide better leakage resilient secret sharing schemes than codes over extension fields. However, AG codes over prime fields cannot be explicitly constructed in polynomial time. So far, only AG codes over extension fields can be constructed in polynomial time \cite{GS96}. In order to overcome this challenge, we have two approaches: (i) consider leakage resilience secret sharing schemes from concatenated AG codes over prime fields (i.e., concatenate AG codes over extension fields with trivial codes over prime fields to get concatenated AG codes over prime fields); (ii) directly study the resilience of secret sharing schemes from AG codes over extension fields.

Note that although the first alternative of constructing leakage resilient secret sharing schemes with constant share size and unbounded number of players may be achieved more easily, in general, the parameter of the ramp secret sharing schemes constructed using the concatenation method can be less flexible. More specifically, ramp secret sharing schemes constructed using the concatenation method generally have larger reconstruction guarantee. Such restriction is not present when we consider a ramp secret sharing scheme arising from AG codes over an extension field. 

It is then natural to consider whether ramp secret sharing schemes constructed using concatenation method may provide a stronger leakage resilience. When considering the two types of ramp secret sharing schemes with similar number of players, number of corrupted players, privacy guarantee and leakage rate, we show that with some reasonable assumptions, the ramp secret sharing scheme defined over extension field can even provide a stronger leakage resilience. Such a comparison result can be found in Lemma~\ref{lem:mt1better}.

This shows that under some reasonable assumptions, in addition of having smaller reconstruction guarantee, an AG code based ramp secret sharing scheme defined over an extension field also provides a better leakage resilience compared to one obtained by concatenation method. This shows that consideration of secret sharing schemes defined over an extension field may provide us with an interesting family of leakage resilient secret sharing schemes over extension fields.

We proceed to state the main result about ramp secret sharing schemes from AG codes over extension fields. We note that for secret sharing schemes over extension fields of characteristic $p,$ to obtain a leakage resilient secret sharing scheme, the leakage rate from each share must be less than $\log p$ bits. This is due to the existence of a distinguishing attack utilizing $\log p$ bits leakage from each share. A more detailed discussion on this attack can be found in Section~\ref{subsec:LSSS}.

\begin{restatable}{mt}{mtone}\label{mt:1}
Let $q=p^2$ for some prime $p$ and $\F_q$ be a finite field of $q$ elements. Then there exists an infinite family of ramp secret sharing schemes that can be explicitly constructed over $\F_q$ in polynomial time with share size $O(1)$ bits for unbounded number of players $N,$ privacy $T$ and reconstruction $R=T+\frac{2N}{\sqrt{q}-1}+O(1)$ such that any of such secret sharing schemes is $(\theta,\mu, \epsilon)$-LL resilient for any $\theta<T$ and $\mu<\log p$ where

\begin{equation*}
\epsilon=\min\left(q^{\left(N-T-\frac{N}{\sqrt{q}-1}\right)}\cdot c_\mu^{T-\theta+1},2^{(N-T-1)(5\mu+1)+\mu}\cdot (c_\mu^\prime)^{2T-N-\theta}\right)
\end{equation*}
with $c_\mu=\frac{2^\mu \sin\left(\frac{\pi}{2^\mu}\right)}{p\sin\left(\frac{\pi}{p}\right)}$ and $c_\mu^\prime=\frac{2^\mu\sin\left(\frac{\pi}{2^\mu}+\frac{\pi}{2^{4\mu}}\right)}{p\sin\left(\frac{\pi}{p}\right)}.$
\end{restatable}

Apart from considering an extension field, an alternative way to reduce the share size requirement with respect to the number of players is by concatenating a secret sharing scheme over an extension field with the trivial code over a base field. Such method results in a secret sharing scheme with a larger number of players each holding a share of smaller size. Applying this method, we have the following main result.

\begin{restatable}{mt}{mttwo}~\label{mt:2}
Let $q$ be a prime and $\F_q$ be a finite field of $q$ elements. Then there exists an infinite family of ramp secret sharing schemes that can be explicitly constructed over $\F_q$ in polynomial time with share size $O(1)$ bits for unbounded number of players $N,$ privacy $T$ and reconstruction $R=\frac{N}{2}+T+\frac{2N}{q-1}+O(1)$ such that any of such secret sharing schemes is $(\theta, \mu,\epsilon)$- LL resilient for any $\theta<T$ and $\mu < \log q$ where
% Let $F/\F_Q$ be a function field of genus $\mathfrak{g}$ with at least $n+2$ distinct $\F_Q$ rational places $\mathtt{P}_\infty, \mathtt{P}_0,\cdots, \mathtt{P}_n.$ Consider a ramp secret sharing scheme $EAGSh_{N,R,T}$ which is obtained by concatenating the AG-code based ramp secret sharing scheme $AGSh_{mP_\infty,\mathcal{P}}$ with the vector space isomorphism $\Pi_{1,2}:\F_Q\rightarrow (\F_q)^2.$ This implies $N=2n,T=m-2\mathfrak{g}$ and $R=n+m+1=\frac{N}{2}+T+\mathfrak{g}+1.$ Let $\theta<T$ and $\mu<\log q$ be positive integers.

%Then the ramp secret sharing scheme $EAGSh_{N,R,T}(s)$ with each share being defined over $\F_{q}$ is $(\theta,\mu,\epsilon)$- LL resilient where
\begin{equation*}
\epsilon=\min\left(q^{\left(N-2T-\frac{N}{q-1}\right)}\cdot c_\mu^{T-\theta+1},2^{(N-\theta-T-1)(5\mu+1)+\mu}\cdot (c_\mu^\prime)^{2T-N}\right)
\end{equation*}
with $c_\mu=\frac{2^\mu \sin\left(\frac{\pi}{2^\mu}\right)}{q\sin\left(\frac{\pi}{q}\right)}$ and $c_\mu^\prime=\frac{2^\mu\sin\left(\frac{\pi}{2^\mu}+\frac{\pi}{2^{4\mu}}\right)}{q\sin\left(\frac{\pi}{q}\right)}.$
\end{restatable}

\begin{rmk}{\rm The reconstruction parameter $R$ in Main Theorem \ref{mt:2} is $\Omega(N)$ larger than that given in Main Theorem \ref{mt:1}. Furthermore, under some reasonable assumptions, the secret sharing schemes from algebraic geometry codes over extension fields given in Main Theorem \ref{mt:1} have stronger leakage resilience property than those given in Main Theorem \ref{mt:2} (The comparison is given in Lemma \ref{lem:mt1better}). This shows a strong motivation to study secret sharing schemes over extension fields.
}\end{rmk}

\subsection{Our Techniques}
The calculation of leakage resilience depends on the statistical distance between the outputs of leakage functions given that the inputs are either a secret sharing using a specific scheme or a set of uniformly and independently sampled strings of the same lengths. In order to facilitate such calculation, we closely follow the proof idea of~\cite{BDIR19}. More specifically, we utilize Fourier Analysis to transform the statistical distance formula to a sum of some Fourier coefficients. When the underlying field is a prime field $\F_p,$ such Fourier coefficients are always pairwise distinct $p$-th roots of unity. Because of this, bounding such sum can then be reduced to bounding the sum of some pairwise distinct $p$-th roots of unity. When we generalize the underlying field to an arbitrary finite field $\F_q$ for some prime power $q=p^w$ where $p$ is a prime and $w$ is a positive integer, the Fourier coefficients are now defined as $\omega_p^{Tr(\cdot)}$ where $Tr(\cdot)$ is the field trace of $\F_{p^w}$ over $\F_p.$ Since field trace function is not injective, although the Fourier coefficients are still $p$-th roots of unity, they may no longer be pairwise distinct. So, instead of bounding the sum of $s$ pairwise distinct $p$-th roots of unity, we need to derive a bound of integer combinations of $p$-th roots of unity with the sum of the coefficients being fixed to $s.$ Such upper bound can be found in Lemma~\ref{lem:newsumpthroot}.

Here we provide some intuition on how to establish such upper bound. Note that each $\omega_p^{Tr(i)}$ has the same length while having different directions. Sum of two of such vectors are maximized when they have the same direction while it decreases as the angle between the two vectors increases. Hence the strategy to maximize the sum is to have as many vectors with the same directions as possible to be summed up. Once such vectors are exhausted, to maximize the sum, we need to choose vectors with the smallest angle with the current sum. Once such vector is chosen, we can again maximize the sum by adding vectors with the same direction as the one we just chose until such vectors are exhausted. We can keep doing this until we have summed up $s$ of such vectors. Lemma~\ref{lem:newsumpthroot} confirmed that such strategy indeed leads to a tight upper bound of the sum.

Having such upper bound, we cannot apply it directly to the ramp secret sharing schemes we are interested in. This is because the adversary learns not only the leak from the shares, but he also learns the full share of some of the players he corrupted. Having such information, the remaining shares no longer follow the distribution of the original secret sharing scheme. Hence, instead of analysing the leakage resilience of the secret sharing scheme itself, we need to consider the leakage resilience for a more general linear or affine codes that we obtain after the corrupted shares are already considered.  In the following, we present the leakage resilience results that applies to any linear codes with leakage defined over its coordinates.

\noindent\textbf{Leakage Resilience of Linear Codes over Arbitrary Finite Fields.} Let $C\subseteq\F_{p^w}^n$ be an $[n,k,d\leq n-k+1]$ code with the dual code $C^\bot$ which is an $[n,n-k,d^\bot\leq k+1]$ code. Let ${\boldsymbol \tau}=(\tau^{(1)},\cdots, \tau^{(n)})$ be any family of leakage functions, each $\tau^{(i)}$ having $\mu(<\log p)$-bit output. Letting $c_\mu=\frac{2^\mu\sin\left(\frac{\pi}{2^\mu}\right)}{p\sin\left(\frac{\pi}{p}\right)}$ and $c_\mu^\prime =\frac{2^\mu\sin\left(\frac{\pi}{2^\mu}+\frac{\pi}{2^{4\mu}}\right)}{p\sin\left(\frac{\pi}{p}\right)},$ the statistical distance between the distribution of ${\boldsymbol \tau}(\mathbf{c})$ when $\mathbf{c}$ is uniformly sampled from $C$ and that when $\mathbf{c}$ is uniformly sampled from $\F_q^n$ is upper bounded by $\frac{1}{2}p^{w(n-k)}c_\mu^{d^\bot}$ and $\frac{1}{2} 2^{(5\mu+1)(n-d^\bot)+\mu}\cdot (c_\mu^\prime)^{2d^\bot-n-2}.$

Having such result, we can then obtain the leakage resilience of our secret sharing schemes by identifying the linear code $C$ such that the distribution of the leak given the knowledge of the corrupted share is equivalent to the distribution of the leakage obtained from the coordinates of codewords in $C.$

\subsection{Organization}
This paper is organized as follows. In Section~\ref{sec:prelim}, we define some notations that will be used throughout the paper and briefly discuss some basic concepts that are useful in our discussion. Section~\ref{sec:FourAn} provides some review on Fourier Analysis and discussion on some results that are essential in our investigation of leakage resilience for secret sharing schemes defined over finite fields. Section~\ref{sec:LRLinCod} discusses and proves the leakage resilience results for general linear codes defined over arbitrary finite fields. This result is then used to provide some leakage resilience property of additive and Shamir's secret sharing schemes that are defined over arbitrary finite fields. Lastly, such result is then used to investigate the local leakage resilience of ramp secret sharing schemes that are defined based on algebraic geometric codes, which can be found in Section~\ref{sec:AGSSS}. Due to page limit, we move some proofs to Supplementary Material.

\section{Notation and Preliminaries}\label{sec:prelim}
\subsection{Notation}\label{subsec:Not}
Throughout the paper, we use $\mathbb{C}$ to denote the field of complex numbers, $\mathbf{i}=\sqrt{-1}\in \mathbb{C}$ the imaginary number and $\mathbb{U}_1=\{x\in \mathbb{C}:\|x\|=1\}.$ For a prime $p,$ let $q$ be a power of $p,$ i.e. $q=p^{w}$ for some positive integer $w.$ For any positive integers $a$ and $b,$ we denote by $\F_q^a$ and $\F_q^{a\times b}$ the set of vectors over $\F_q$ of length $a$ and the set of matrices over $\F_q$ with $a$ rows and $b$ columns respectively.

Let $S$ be any finite set. We use $x\leftarrow S$ to denote that $x$ is sampled with uniform distribution over the elements of $S.$ For any positive integer $n,$  we denote by $\mathcal{U}_n,$ the uniform distribution over $\F_q^n.$ Furthermore, consider two distributions $f$ and $g$ over $S.$ We define the statistical distance between the two distributions $SD(f,g)=\frac{1}{2}\sum_{s\in S} |f(s)-g(s)|.$ We write $f\equiv g$ if $SD(f,g)=0.$ Furthermore, for any $\epsilon>0,$ we say $f\approx_\epsilon g$ if $SD(f,g)\leq \epsilon.$

For any function $f,$ we denote by $\mathbb{E}_{x\leftarrow S}(f(x))$ the expectation of $f(x)$ when $x$ is sampled uniformly at random from the finite set $S.$ In other words, $\mathbb{E}_{x\leftarrow S}(f(x))$ $=\frac{1}{|S|}\sum_{x\in S}f(x).$ We use $[n]$ to represent the set of positive integers of value at most $n.$ That is, $[n]=\{1,\cdots, n\}.$

For any polynomials $p(x)$ and $f(x),$ we define $p(x)^a||f(x)$ if $p(x)^a| f(x)$ but $p(x)^{a+1}\not|f(x).$

Given a set $T$ of size $n$ where its elements are indexed by integers from $1$ to $n, T=\{t_1,\cdots, t_n\}$ and $S_T\subseteq T,$ we denote by $\mathbf{v}^{(S_T)}$ the vector of length $|S_T|$ containing $v_i$ for any $t_i\in S_T.$ In particular, this also applies when $T=[n]$ and $t_i=i.$

\subsection{Linear Secret Sharing Scheme}\label{subsec:LSSS}

Let $\mathtt{U}=\{U_1,\cdots, U_n\}$ be a finite set of players. A forbidden set $\mathcal{F}$ is a family of subsets of $\mathtt{U}$ such that for any $A\in \mathcal{F}$ and $A'\subseteq A,$ we must have $A'\in \mathcal{F}.$ For any $t<n,$ we define $\mathcal{F}_{t,n}$ to be the forbidden set containing all subsets of $\mathtt{U}$ of size at most $t.$ On the other hand, a qualified set $\Gamma$ is a family of subsets of $\mathtt{U}$ such that for any $B\in \Gamma$ and $B\subseteq B',$ we must have $B'\in \Gamma.$ For any $r\leq n,$ we define $\Gamma_{r,n}$ to be the qualified set containing all subsets of $\mathtt{U}$ of size at least $r.$ For any forbidden set $\mathcal{F}$ and a qualified set $\Gamma$ over $\mathtt{U}$ such that $\mathcal{F}\cap \Gamma=\emptyset,$ the pair $(\mathcal{F},\Gamma)$ is called an access structure.

A secret sharing scheme with access structure $(\mathcal{F},\Gamma)$ over $\F_q$ on $\mathtt{U}$ is a pair of functions $(Share,Rec)$ where $Share$ is a probabilistic function that calculates the random shares for the $n$ players given the secret. For any secret $s\in \F_q,$ if $(\mathbf{s}_1,\cdots, \mathbf{s}_n)=Share(s),$ for any $A\subseteq \mathtt{U},$ we denote by $\mathbf{s}^{A}=(\mathbf{s}_i)_{U_i\in A},$ the vector containing the shares of all players $U_i\in A.$ On the other hand, $Rec$ accepts shares from a set of players in $\mathtt{U}$ and attempt to recover the original secret that satisfies the following requirements:
\begin{enumerate}
\item For any $B\in \Gamma,$ given the shares of $U_i\in B, Rec$ returns the original secret.
\item For any $A\in \mathcal{F},$ the shares of $U_i\in A$ does not give any information regarding the secret. That is, the probability that the function $Share$ outputs the given shares to $U_i$ in $B$ is independent of the value of the secret.
\end{enumerate}

A secret sharing scheme with access structure $(\mathcal{F},\Gamma)$ such that $\mathcal{F}_{t,n}\subseteq \mathcal{F}$ and $\Gamma_{r,n}\subseteq \Gamma$ for some $0<t<r<n$ is called a ramp secret sharing scheme providing $t$ privacy and $r$ reconstruction. If $r=t+1,$ we call it a threshold secret sharing scheme.

A linear secret sharing scheme (LSSS) with access structure $(\mathcal{F},\Gamma)$ over $\F_q$ on $\mathtt{U}$ is defined as follows. Fix a positive integer $\mathtt{m}$ and $V_1,\cdots, V_n$ subspaces of $\F_q^{\mathtt{m}}.$ We also fix $\mathbf{u}\in \F_q^{\mathtt{m}}\setminus\{0\}$ which can be set to be $(1,0,0,\cdots,0)$ without loss of generality. For any $A\subseteq [n],$ define $V_A=\sum_{i\in A}V_i$ which is the smallest subspace containing $V_i$ for all $i\in A.$ We further fix $V_{i}^\ast$ a basis of $V_i.$ An LSSS $(Share,Rec)$ is defined as follows. Let $s\in \F_q$ be the secret. Then the $Share$ function starts by choosing a random linear map $\phi:\F_q^{\mathtt{m}}\rightarrow \F_q$ such that $\phi(\mathbf{u})=s.$ The share for player $U_i$ is then defined as $\mathbf{s}_i=\phi(V_i^\ast)=\{\phi(\mathbf{x}):\mathbf{x}\in V_i^\ast\}.$

Now we define the function $Rec.$ Let $B\subseteq \mathtt{U}$ be such that $\mathbf{u}\in V_B.$ Due to the linearity of $\phi$ and the fact that $\mathbf{u}\in V_B,$ there exists a vector $\mathbf{w}\in \F_q^{(\sum_{U_i\in B}|V_i^\ast|)}$ such that $s=\phi(\mathbf{u})=\mathbf{w}\cdot (\mathbf{s}_i)_{U_i\in B}.$ This further proves that the qualified set $\Gamma$ of this LSSS is $\Gamma=\{B\subseteq\mathtt{U}:\mathbf{u}\in V_B\}.$ A simple algebraic manipulation tells us that $\mathcal{F}=\{A\subseteq\mathtt{U}:~\exists\kappa:\F_q^{\mathtt{m}}\rightarrow\F_q, \kappa(\mathbf{x})=0~\forall \mathbf{x}\in V_A$ and $\kappa(\mathbf{u})=1\}.$ For the remainder of the paper, we only consider linear secret sharing schemes.

To model leakage-resilient secret sharing, first we discuss the local leakage model as proposed in~\cite{GK18} and~\cite{BDIR19}. In short, the adversary can provide an arbitrary independent leakage functions for each party with a fixed output length that will provide the leakage of information from each of the share to the adversary. We note that this information is provided in addition to the capability of the adversary to control some number of players. The following definitions formalize the concept of leakage function and leakage resilience.

\begin{defn}[Leakage Function]\label{def:LF} Let ${\boldsymbol\tau}=(\tau_1,\cdots, \tau_n)$ be a vector of functions where for each $i=1,\cdots, n, \tau_i:(\F_q)^{|V_i^\ast|}\rightarrow \F_2^{\mu}.$ Then for a secret sharing $(\mathbf{s}_1,\cdots,\mathbf{s}_n)=Share(s)$ of a secret $s\in \F_q,$ define $(b_1,\cdots, b_n)={\boldsymbol\tau}(\mathbf{s}_1,\cdots, \mathbf{s}_n)=(\tau_1(\mathbf{s}_1),\cdots, \tau_n(\mathbf{s}_n)).$ Given a set of players $\Theta\subseteq \mathtt{U}$ and $\mu$ bits output leakage function ${\boldsymbol\tau}.$ We define the information learned by the adversary on a secret sharing $\mathbf{s}=(\mathbf{s}_1,\cdots, \mathbf{s}_n)=Share(s)$ as
\[Leak_{\Theta,{\boldsymbol\tau}}(\mathbf{s})\triangleq\left(\mathbf{s}^{(\Theta)},(\tau_i(\mathbf{s}_i))_{U_i\in\mathtt{U}\setminus\Theta}\right).\]
\end{defn}

Next, we define the concept of local leakage resilient or $LL$ resilience for short.

\begin{defn}[Local Leakage Resilience]\label{def:LLR}
Let $\Theta\subseteq\mathtt{U}$ be a set of players. A secret sharing scheme $(Share,Rec)$ is said to be $(\Theta,\mu,\epsilon)$-local leakage resilient (or $(\Theta,\mu,\epsilon)$-LL resilient) if for any leakage function family ${\boldsymbol \tau}=(\tau_1,\cdots, \tau_n)$ where each $\tau_j$ has a $\mu$-bit output and for every pair of secrets $s_0,s_1\in \F_q,$ we have
\[ SD\left(\left\{Leak_{\Theta,{\boldsymbol\tau}}(\mathbf{s}): \mathbf{s}=Share(s_0)\right\},\left\{Leak_{\Theta,{\boldsymbol\tau}}(\mathbf{s}): \mathbf{s}=Share(s_1)\right\}\right)\leq \epsilon.\]

For a positive integer $\theta\leq n,$ we say a secret sharing scheme $(Share,Rec)$ is $(\theta,\mu,\epsilon)$-LL resilient if for any $\Theta\subseteq \mathtt{U},|\Theta|\leq \theta,$ it is $(\Theta,\mu,\epsilon)$-LL resilient.
\end{defn}

We conclude this subsection by providing an attack that shows that any linear secret sharing scheme over $\F_{p^{w}}$ does not provide any leakage resilience if the leak is beyond $\log p$ bits. The attack provided below is a generalization on the attacks on additive and Shamir's secret sharing scheme over fields of characteristic two. Recall that for any $\alpha\in \F_q,$ by fixing an $\F_p$-basis of $\F_q,~\{\lambda_1,\cdots, \lambda_{w}\},$ we can see $\alpha$ as a vector of length $w$ over $\F_p.$ If $\alpha=\sum_{i=1}^{w}\lambda_i \alpha_i$ for some $\alpha_i\in \F_p,$ we define $\varphi_p(\alpha)=\alpha_1\in \F_p.$

Consider a linear secret sharing scheme with access structure $(\mathcal{F},\Gamma)$ over $\F_q$ on $\mathtt{U}$ such that $\Gamma\neq\emptyset.$ Then it is easy to see that $\mathtt{U}\in \Gamma.$ Set $\mathtt{m}, V_1,\cdots, V_n$ subspaces of $\F_q^{\mathtt{m}}$ with their corresponding bases $V_1^\ast,\cdots, V_n^\ast$ and $\mathbf{u}=(1,0,\cdots,0)\in \F_q^{\mathtt{m}}\setminus\{0\}$ as defined above. For $i=1,\cdots, n,$ let $V_i^\ast = \{\mathbf{v}_{i,1},\cdots, \mathbf{v}_{i,d_i}\}.$ Recall that by the analysis above, since $\mathtt{U}\in \Gamma,$ there exists $\mathbf{w}^\ast=(w_{1,1},\cdots, w_{n,d_n})\in \F_q^{\sum_{i=1}^n |V_i^\ast|}$ such that $\mathbf{u}=\mathbf{w}^\ast\cdot (\mathbf{v}_{1,1},\cdots, \mathbf{v}_{1,d_1},\cdots,\mathbf{v}_{n,1},\cdots, \mathbf{v}_{n,d_n})^T.$
For $i=1,\cdots,n,$ we define $\mathbf{u}_i=(w_{i,1},\cdots, w_{i,d_i})\cdot (\mathbf{v}_{i,1},\cdots, \mathbf{v}_{i,d_i})^T.$ Furthermore, we define $s_i=(w_{i,1},\cdots,w_{i,d_i})\cdot \mathbf{s}_i.$ It is easy to see that $s_i$ can be locally computed by player $U_i$ and  $s=\sum_{i=1}^n s_i.$ Note that due to the $\F_p$-linearity of $\varphi_p,$ we have $\varphi_p(s)\equiv \sum_{i=1}^n \varphi_p(s_i)\pmod{p}.$ Hence if we define $\tau_i(\mathbf{s}_i)=\varphi_p(s_i)$ where $s_i$ is defined as discussed above, we have a local leakage function that outputs $\log(p)$-bit leakage from each shares and these leakages can be used to distinguish the value of $\varphi_p(s)$ with probability $1.$ This shows that any linear secret sharing schemes that is defined over $\F_{p^{w}}$ is not $(\theta,\log(p),\epsilon)$ for any $\epsilon<1.$ Because of this, for the remainder of this paper, we always assume that the length of the leakage $\mu$ is less than $\log(p).$

\subsection{Linear Codes}\label{subsec:LinCode}
In this section, we briefly discuss the concept of linear codes and in particular, a family of linear codes called Algebraic Geometric code or AG code for short.

\begin{defn}[Linear Codes]
Let $n,k,d$ be non-negative integers such that $d$ and $k$ are at most $n.$ A linear code $C$ over $\F_q$ with parameter $[n,k,d]$ is a subspace $C\subseteq \F_q^n$ of dimension $k$ such that for any non-zero $\mathbf{c}\in C\setminus\{\mathbf{0}\}, wt_H(\mathbf{c})\geq d$ where for any vector $\mathbf{x}=(x_1,\cdots, x_n), wt_H(\mathbf{x})$ is defined to be the Hamming weight of $\mathbf{x},$ i.e., $wt_H(\mathbf{x})=|\{i:x_i\neq 0\}|.$ By the Singleton bound, we have the relation $d\leq n-k+1;$ a code that satisfies this bound with an equality is called a Maximum Distance Separable code or MDS code for short.

A linear $[n,k,d]$ code $C$ can be represented by its generator matrix $G\in \F_q^{k\times n}.$ So given $G,$ we have $C=\{\mathbf{x}\cdot G: \mathbf{x}\in \F_q^k\}.$

Given a linear code $C,$ its dual code $C^\bot$ is defined to be the dual subspace of $C$ over $\F_q^n.$ That is, $C^\bot=\{\mathbf{x}\in \F_q^n: \langle \mathbf{c},\mathbf{x}\rangle =0~\forall \mathbf{c}\in C\}$ where $\langle \cdot,\cdot\rangle$ denotes the inner product operation. Then, $C^\bot$ is an $[n,n-k,d^\bot\leq k+1]$ code. A parity check matrix $H\in \F_q^{(n-k)\times n}$  of $C$ is a generator matrix of $C^\bot.$
\end{defn}

In the following, we provide the characterization of a linear code based on its generator matrix and its parity check matrix.

\begin{prop}\label{prop:distprop}
An $[n,k,d]$ linear code $C$ has minimum distance $d$ if and only if every set of $d-1$ columns of its parity check matrix $H\in \F_p^{(n-k)\times n}$ are linearly independent and there exists a set of $d$ columns of $H$ that is linearly dependent. Furthermore, for its generator matrix $G\in \F_p^{k\times n},$ the minimum distance of $C$ is $d$ if and only if any submatrix $G'\in \F_p^{k\times(n-d+1)}$ has full rank while there exists a submatrix $\hat{G}\in \F_p^{k\times (n-d)}$ that does not have a full rank.
\end{prop}

\subsection{Algebraic Geometric Code}\label{subsec:AGCode}
In this section, we briefly discuss algebraic-geometric codes (AG codes for short) and their properties. Before we discuss the construction of AG codes, first we provide a brief discussion on algebraic function field. For a complete discussion on algebraic function field and the construction of AG code, see~\cite{St08}.

Let $\mathbb{K}$ be a field. We say $\mathbb{F}$ is an algebraic function field over $\mathbb{K}$ in one variable if $\mathbb{F}$ is a finite algebraic extension of $\mathbb{K}(x)$ for some element $x\in \mathbb{F}$ that is transcendental over $\mathbb{K}.$ The algebraic closure $\overline{\mathbb{K}}$ of $\mathbb{K}$ in $\mathbb{F}$ contains all the elements in $\mathbb{F}$ that is algebraic over $\mathbb{K}.$ The field $\mathbb{K}$ is called the full constant field of $\mathbb{F}$ if $\overline{\mathbb{K}}=\mathbb{K}.$

We define a map $\varphi:\mathbb{F}\rightarrow\mathbb{Z}\cup\{\infty\}$ as a discrete valuation of $\mathbb{F}/\mathbb{K}$ if: (i) $\varphi(0)=\infty$ and $\varphi(x)=0$ for any $x\in \mathbb{K}\setminus\{0\},$ (ii) For any $x,y\in \mathbb{F}, \varphi(xy)=\varphi(x)+\varphi(y)$ and $\varphi(x+y)\geq \min\{\varphi(x),\varphi(y)\}$ and (iii) $\varphi^{-1}(1)\neq \emptyset.$
%\begin{enumerate}
%\item For $x\in \mathbb{K},$ we have
%\begin{equation*}
%\varphi(x)=\left\{
%\begin{array}{cc}
%\infty, &\mathrm{~if~}x=0\\
%0, &\mathrm{~otherwise.}
%\end{array}
%\right.
%\end{equation*}
%\item $\forall x,y\in \mathbb{F}, \varphi(xy)=\varphi(x)+\varphi(y)$ and $\varphi(x+y)\geq \min\{\varphi(x),\varphi(y)\}.$
%\item $\varphi^{-1}(1)\neq \emptyset.$
%\end{enumerate}

Any discrete valuation $\varphi$ defines a valuation ring $\mathcal{O}\triangleq\{x\in \mathbb{F}:\varphi(x)\geq 0\}$ which is a local ring with its maximal ideal $P\triangleq \{x\in \mathbb{F}:\varphi(x)>0\},$ which is called a place. The collection of all places in $\mathbb{F}$ is denoted by $\mathbb{P}_\mathbb{F}.$ The discrete valuation and the valuation ring corresponding to a place $P$ are denoted by $\varphi_P$ and $\mathcal{O}_P$ respectively. By definition, $F_P\triangleq \mathcal{O}_P / P$ is a finite field extension of $\mathbb{K}.$ The degree of $P,$ denoted by $\deg(P)$ is defined to be the extension degree $[F_P:\mathbb{K}].$ We say $P$ is a rational place if $\deg(P)=1.$  For any $x\in \mathcal{O}_P,$ we define $x(P)\in F_P,$ the residue class of $x$ modulo $P.$ So if $P$ is a rational place, $x(P)\in \mathbb{K}.$

For a non-zero $x\in \mathbb{F},$ we denote its principal divisor by ${\rm div}(x)$ which is defined as ${\rm div}(x)=\sum_{P\in \mathbb{P}_\mathbb{F}}\varphi_P(x) P.$ We also define ${\rm div}(x)_0=\sum_{P: \varphi_P(x)>0} \varphi_P(x) P$ and ${\rm div}(x)_\infty = -\sum_{P:\varphi_P(x)<0} \varphi_P(x) P$ the zero and pole divisors of $x$ respectively. The divisor group, ${\rm Div}(\mathbb{F})$ is the free Abelian group generated by all elements of $\mathbb{P}_{\mathbb{F}}.$ An element $D\in {\rm Div}(\mathbb{F})$ is called a divisor of $\mathbb{F}$ and it can be written as $D=\sum_{P\in \mathbb{P}_{\mathbb{F}}} n_P(D) P$ where $n_P(D)\in \mathbb{Z}$ and $n_P(D)=0$ for all but finitely many $P\in \mathbb{P}_{\mathbb{F}}.$ The finite set of all places $P\in \mathbb{P}_{\mathbb{F}}$ where $n_P(D)\neq 0$ is called its support, denoted by ${\rm Supp}(D).$ We denote by $0$ a divisor of $\mathbb{F}/\mathbb{K}$ with $n_P(0)=0$ for all $P\in \mathbb{P}_\mathbb{F}.$ For any two divisors $D$ and $D',$ we say $D\geq D'$ if for any $P\in \mathbb{P}_{\mathbb{F}},$ we have $n_P(D)\geq n_P(D').$

Next we define the Riemann-Roch space associated to a divisor. For any divisor $D$ of $\mathbb{F}/\mathbb{K},$ the Riemann-Roch space $\mathcal{L}(D)$ associated with $D$ is defined as $\mathcal{L}(D)\triangleq \{x\in \mathbb{F}\setminus\{0\}: {\rm div}(x)+D \geq 0\} \cup \{0\}.$
It can be shown that $\mathcal{L}(D)$ is a finite dimensional $\mathbb{K}$-vector space with dimension $\dim_{\mathbb{K}}(\mathcal{L}(D))\geq \deg(D)+1-\mathfrak{g}$ where $\mathfrak{g}$ is the genus of $\mathbb{F}$ and equality holds if $\deg(D)\geq 2\mathfrak{g}-1$~\cite{St08}.

Consider the case when $\mathbb{K}=\F_q$ and $\F=\F_q(x).$ Then every discrete valuation of $\mathbb{F}/\mathbb{F}_q$ is either $\varphi_{(p(x))}$ for some irreducible polynomial $p(x)$ or $\varphi_\infty$ where for any $f/g\in \mathbb{F}_q(x), \varphi_{(p(x))}(f/g) = a-b$ where $p(x)^a||f$ and $p(x)^b||g$ and $\varphi_\infty(f/g)=\deg(g)-\deg(f).$

Now we are ready to construct algebraic geometric codes. Let $q=p^w$ for some prime $p$ and a positive integer $w.$ Assume that $\mathbb{F}/\F_q$ is a function field of genus $\mathfrak{g}$ and at least $n+1$ pairwise distinct rational places. Label the $n+1$ points as $P_1,\cdots, P_n,Q$ and $\mathcal{P}=(P_1,\cdots, P_n).$ Set $m\geq 2\mathfrak{g}-1$ and $D=mQ.$ We define the AG code $C(D,\mathcal{P})$ as $C(D,\mathcal{P})=\{(x(P_1),\cdots, x(P_n)):x\in \mathcal{L}(D)\}.$ Note that for any $i=1,\cdots, n,$ we have $n_{P_i}(x)\geq 0,$ which implies that $x\in \mathcal{O}_P$ making $x(P_i)$ to be a well defined element of $\F_q.$ Hence $C(D,\mathcal{P})\subseteq \F_q^n.$ Furthermore, $C(D,\mathcal{P})$ is an $[n,k=m-\mathfrak{g}+1,d\geq n-m]-$ linear code while its dual code $C^\bot$ is an $[n,n-m+\mathfrak{g}-1,d^\bot\geq m-2\mathfrak{g}+2]-$ linear code~\cite[Section $2.2$]{St08}.

We can construct a ramp secret sharing scheme based on algebraic-geometric code~\cite{CC06}. The construction is done in the following manner. First, let $F/\mathbb{F}_q$ be a function field of genus $\mathfrak{g}$ and at least $n+2$ pairwise distinct rational places, denoted by $P_0,\cdots, P_n,Q.$ Fix the privacy level $t$ such that $1\leq t<n-2\mathfrak{g}.$ Set $m=t+2\mathfrak{g}$ and set $D=(2\mathfrak{g}+t)\cdot Q.$ For any secret $s\in \F_q,$ the $Share$ function works as follows. First, we choose $f\in \mathcal{L}(D)$ at random conditioned on $f(P_0)=s.$ We then define the share for $U_j$ to be $s_i=f(P_i)\in \F_q.$

Let $(\mathcal{F},\Gamma)$ be the access structure of the secret sharing scheme defined above. Then it can be shown~\cite{CC06} that $\mathcal{F}_{t,n}\subseteq \mathcal{F}$ and $\Gamma_{2\mathfrak{g}+t+1,n}\subseteq \Gamma.$ Hence the secret sharing scheme defined above is a ramp secret sharing scheme providing $t$ privacy and $r=2\mathfrak{g}+t+1$ reconstruction.

We note that such Riemann-Roch space along with the corresponding AG codes and the ramp secret sharing schemes defined over it is shown to exist for any prime powers $q=p^w$~\cite{Ser83,NX01}. However, such codes are only shown to be constructable with polynomial time when $q$ is a square of a prime or $w=2$~\cite{GS96}. Furthermore, when we consider $q$ is a square of a prime, we can have up to approximately $\frac{\mathfrak{g}}{\sqrt{q}-1}$ rational points.

Next we provide the surjectivity of the $Share$ function.

\begin{lem}\label{lem:surj}
Let $F/\F_q$ be a function field of genus $\mathfrak{g}$ with $P_{\infty},P_1,\cdots, P_\ell$ its $\ell+1$ pairwise distinct rational points. Consider the linear map $\varphi:\mathcal{L}(mP_\infty)\rightarrow \F_q^\ell$ such that $\varphi(f)=(f(P_1),\cdots, f(P_\ell))$ for any $f\in \mathcal{L}(mP_\infty).$ If $m\geq 2\mathfrak{g}+\ell-1,$ then $\varphi$ is surjective.
\end{lem}

\begin{proof}
Obviously,  $Im(\varphi)$ is a subspace of $\F_q^\ell.$ We aim to prove that they have the same dimension, $\ell.$ First we consider the kernel of $\varphi.$

\begin{eqnarray*}
f\in Ker(\varphi) &\Leftrightarrow & f(P_i)=0~\forall i=1,\cdots, \ell, v_{P_\infty}(f)\geq -m\\
&\Leftrightarrow & v_{P_i}(f)\geq 1~\forall i=1,\cdots, \ell, v_{P_\infty}(f)\geq -m\\
&\Leftrightarrow & f\in \mathcal{L}(mP_{\infty}-P_1-\cdots-P_\ell)
\end{eqnarray*}

So $Ker(\varphi)=\mathcal{L}(mP_{\infty}-P_1-\cdots-P_\ell).$ By rank-nullity theorem, we obtain that $dim(Im(\varphi))=dim(\mathcal{L}(mP_\infty))-dim(Ker(\varphi)).$ Note that $deg(\mathcal{L}(mP_\infty))=m$ and $deg(\mathcal{L}(mP_\infty-P_1-\cdots-P_\ell)=m-\ell.$ By assumption of $m,$ we have $m\geq m-\ell\geq 2g-1.$ So by Riemann's theorem~\cite[Theorem $1.4.17$]{St08}, $dim(\mathcal{L}(mP_\infty))=deg(\mathcal{L}(mP_\infty))+1-\mathfrak{g}=m+1-\mathfrak{g}$ and $dim(\mathcal{L}(mP_\infty-P_1-\cdots-P_\ell))=deg(\mathcal{L}(mP_\infty-P_1-\cdots-P_\ell))+1-\mathfrak{g}=m-\ell+1-g.$ Hence $dim(Im(\varphi))=m+1-\mathfrak{g}-(m-\ell+1-\mathfrak{g})=\ell.$ This shows that $Im(\varphi)=\ell=dim(\F_q^\ell),$ completing the proof.
\end{proof}

Let $C=C(D,\mathcal{P})\subseteq \F_q^n$ for some $q=p^{w}.$ Now we identify $\mathbb{F}_{p^{w}}$ with $\mathbb{F}_p^{w}$ via an $\mathbb{F}_p$-isomorphism $\Pi.$ Then
\begin{equation*}
\hat{C}=\Pi(C)=\left\{(\Pi(c_1)\|\Pi(c_2)\|\cdots\|\Pi(c_n)): (c_1,\cdots, c_n)\in C\right\}\subseteq \F_p^{w n}
\end{equation*}
is a $p$-ary $[wn,k=w m-w \mathfrak{g}+w ,d\geq n-m]$-linear code with dual code $\hat{C}^\bot$ a $p$-ary $[w n,w n-w m+w \mathfrak{g}-w ,d^\bot\geq m-2\mathfrak{g}+2].$ Furthermore, the corresponding secret sharing scheme provides $t=m-2\mathfrak{g}$ privacy and $(w -1)n+2\mathfrak{g}+t+1$ reconstruction.

\section{Fourier Analysis}\label{sec:FourAn}
In this section, we provide a brief discussion on Fourier coefficients of a function along with some of the properties that are useful in the discussion later. For a more complete discussion, see~\cite{Gr07}.

Let $\mathbb{G}$ be any finite Abelian group. A character $\chi:\mathbb{G}\rightarrow\mathbb{U}_1$ is a group homomorphism between the group $\mathbb{G}$ and the multiplicative group $\mathbb{U}_1.$ That is, for any $a,b\in \mathbb{G}, \chi(a+b)=\chi(a)\cdot \chi(b).$ Let $\hat{\mathbb{G}}$ be the set of characters of $\mathbb{G}.$ Then equipped with point-wise product operation, $\hat{\mathbb{G}}$ is a group that is isomorphic to $\mathbb{G}.$ So we can write any element of $\hat{\mathbb{G}}$ by $\chi_g$ for some $g\in \mathbb{G}$ where the correspondence is done using a fixed isomorphism between $\mathbb{G}$ and $\hat{\mathbb{G}}.$

Note that we are interested in the case when $\mathbb{G}=\F_q$ where $q=p^{w }$ where $p$ is a prime number and $w$ is a positive integer. We denote by $\omega_p=e^{\frac{2\pi}{p}\mathbf{i}}\in \mathbb{C}$ the $p$-th root of unity.  Let $\alpha\in \F_q.$ It can be shown that $\chi_\alpha\in \widehat{\F_q}$ is defined to be $\chi_\alpha(x)=\omega_p^{{\rm Tr}_{\F_q/\F_p}(\alpha\cdot x)}$ where ${\rm Tr}_{\F_q/\F_p}(y)$ is the field trace of $\F_q$ over $\F_p.$ That is, ${\rm Tr}_{\F_q/\F_p}(y)=\sum_{i=0}^{w -1}y^{p^i}.$

\begin{defn}[Fourier Coefficients]
For functions $f:\mathbb{G}\rightarrow\mathbb{C},$ the Fourier basis is composed of the group of characters $\chi_g\in \hat{\mathbb{G}}.$ Then the Fourier coefficient $\hat{f}(\chi)$ corresponding to the character $\chi\in \hat{\mathbb{G}}$ is defined to be
\begin{equation*}
\hat{f}(\chi) = \mathbb{E}_{x\leftarrow \mathbb{G}}[f(x)\cdot \chi(x)]\in \mathbb{C}.
\end{equation*}
\end{defn}

Recall that $\hat{\mathbb{G}}$ is isomorphic to $\mathbb{G}.$ So we can identify an element $\alpha$ of $\mathbb{G}$ with the character $\chi_\alpha\in \hat{\mathbb{G}}.$ To simplify the notation, for any function $f:\mathbb{G}\rightarrow \mathbb{C}$ and $\alpha\in \mathbb{G},$ instead of writing $\hat{f}(\chi_\alpha),$ we write $\hat{f}(\alpha).$ Next we provide some existing properties of Fourier Transform.

\begin{lem}\label{lem:PIFIF}
Let $\mathbb{G}$ be a finite Abelian group and $\hat{\mathbb{G}}$ be its group of characters. We further let $f,g:\mathbb{G}\rightarrow\mathbb{C}$ be two functions. Then
\begin{enumerate}
\item (Parseval's Identity) We have
\begin{equation*}
\mathbb{E}_{x\leftarrow\mathbb{G}}\left[f(x)\cdot\overline{g(x)}\right]=\sum_{\chi\in \hat{\mathbb{G}}}\hat{f}(\chi)\cdot\overline{\hat{g}(\chi)}
\end{equation*}
where $\hat{f}(x)$ is the Fourier coefficient of $f$ corresponding to a character $\chi,$ i.e. $\hat{f}(\chi)=\mathbb{E}_{x\in \mathbb{G}}[f(x)\cdot \chi(x)]\in \mathbb{C}.$
In particular, $\|f\|_2=\|\hat{f}\|_2$ where $\|f\|_2^2=\mathbb{E}_{x\leftarrow\mathbb{G}}\left[|f(x)|^2\right]$ and $\|\hat{f}\|_2^2 = \sum_{\chi\in \hat{\mathbb{G}}}\left|\hat{f}(x)\right|^2.$
\item (Fourier Inversion Formula) For any $x\in \mathbb{G}, f(x)=\sum_{\chi\in \hat{\mathbb{G}}} \hat{f}(x)\cdot\overline{\chi(x)}.$
\end{enumerate}
\end{lem}
We proceed by providing some analysis on some sums of roots of unity when $\mathbb{G}=\F_{p^{w }}$ for a prime $p.$ Suppose that $\omega_p=e^{\frac{2\pi \mathbf{i}}{p}}\in \mathbb{C}$ is a $p$-th root of unity. For any $S\subseteq \F_{p^{w }},$ we define $\omega_p^S\triangleq\sum_{i\in S} \omega_p^{Tr(i)}.$ Lemma~\ref{lem:newsumpthroot} provides an upper bound of such sums which only depends on the size of $S.$
\begin{lem}\label{lem:newsumpthroot}
Let $S\subseteq\F_{p^{w }}$ of size $s\leq p^{w }-1$ and set $s=s_1 (p^{{w }-1})+s_2$ for some $0\leq s_1\leq p-1$ and $0\leq s_2\leq p^{{w }-1}-1.$ For any $i=0,\cdots,p-1,$ we set $T_i=\{x\in \F_{p^{w }}: Tr(x)=i\}.$ Recall that $|T_i|=p^{{w }-1}$ for any $i.$ We set $T_{s_1}^\ast$ to be any subset of $T_{s_1}$ of size $s_2$ and $S^\star=\left(\bigcup_{i=0}^{s_1-1} T_i\right)\cup T_{s_1}^\ast.$ Then
\begin{equation*}
|\omega_p^S|\leq |\omega_p^{S^\star}|=\left|s_2 \sum_{i=0}^{s_1} \omega_p^{i}+ (p^{{w }-1}-s_2)\sum_{i=0}^{s_1-1}\omega_p^i\right|\leq p^{{w }-1}\cdot\frac{\sin(\pi s/p^{w })}{\sin(\pi/p)}.
\end{equation*}
\end{lem}

\begin{proof}
Note that the equality can be easily verified by the definition of $S^\ast.$ Using triangle inequality and further algebraic manipulation, we get
\begin{eqnarray*}
\left|s_2 \sum_{i=0}^{s_1} \omega_p^{i}+ (p^{{w }-1}-s_2)\sum_{i=0}^{s_1-1}\omega_p^i\right|&\leq& \left|s_2 \sum_{i=0}^{s_1} \omega_p^{i}\right|+ \left|(p^{{w }-1}-s_2)\sum_{i=0}^{s_1-1}\omega_p^i\right|\\
&=&s_2 \frac{\left|\omega_p^{s_1+1}-1\right|}{\left|\omega_p-1\right|}+ (p^{{w }-1}-s_2) \frac{\left|\omega_p^{s_1}-1\right|}{\left|\omega_p-1\right|}\\
&=& s_2 \frac{\sin(\pi (s_1+1)/p)}{\sin(\pi/p)}+ (p^{{w }-1}-s_2)\frac{\sin(\pi s_1/p)}{\sin(\pi/p)}\\
&=& \frac{p^{{w }-1}}{\sin(\pi/p)}\left(\frac{s_2}{p^{{w }-1}}\cdot \sin \left(\frac{\pi}{p}\cdot(s_1+1)\right)\right.\\
&&+\left.\frac{p^{{w }-1}-s_2}{p^{{w }-1}} \cdot \sin \left(\frac{\pi}{p}\cdot s_1\right)\right)
\end{eqnarray*}
Note that $0\leq \pi s_1/p < \pi(s_1+1)/p\leq \pi.$  By the concavity of $\sin(x)$ for $x\in[0,\pi],$ we have that for any $0\leq x<y\leq \pi$ and $\alpha\in [0,1], \sin(\alpha x+(1-\alpha) y)\geq \alpha \sin(x)+(1-\alpha)\sin(y).$ Setting $\alpha=\frac{p^{{w }-1}-s_2}{p^{e-1}}\in[0,1),$ noting that $\alpha s_1 +(1-\alpha) (s_1+1)= \frac{p^{{w }-1}s_1+s_2}{p^{{w }-1}}=\frac{s}{p^{e-1}},$ we can have the term in the last equation $\frac{p^{{w }-1}}{\sin(\pi/p)}\left(\frac{s_2}{p^{{w }-1}}\cdot \sin \left(\frac{\pi}{p}\cdot(s_1+1)\right)+\frac{p^{{w }-1}-s_2}{p^{{w }-1}} \cdot \sin \left(\frac{\pi}{p}\cdot s_1\right)\right)$ to be at most
\[
\frac{p^{{w }-1}}{\sin(\pi/p)}\cdot \sin\left(\frac{\pi}{p}\cdot\frac{s}{p^{{w }-1}}\right)=\frac{p^{{w }-1}\sin(\pi s/p^{w })}{\sin(\pi/p)}.
\]

It remains to prove the first inequality of the claim. Suppose that $S\subseteq \F_{p^{w }}$ of size $s=s_1\cdot p^{{w }-1}+s_2\leq p^{w }-1$ for some $0\leq s_1\leq p-1$ and $0\leq s_2\leq p^{{w }-1}-1$ has the largest value for $|\omega_p^S|$ and set $\xi=|\omega_p^S|.$ That is, $\xi=|\omega_p^S|\geq |\omega_p^{S^\star}|.$ For any $i=0,\cdots, p-1,$ define $S_i=S\cap T_i$ and $|S_i|=s^{(i)}\geq 0$ for each $i$ and $\sum_{i=0}^{p-1} s^{(i)}=s.$ Then $\omega_p^{S}=\sum_{i=0}^{p-1} s^{(i)}\omega_p^{i}.$ First we consider the case when $s\leq p^{w-1}.$ For any $x,y\in \mathbb{C}$ of length $1,$ assuming that $\theta\in[0,\pi]$ is the angle between $x$ and $y,$ without loss of generality, we can write $y=x\cdot e^{\mathbf{i}\theta}.$ Then $|x+y|=|1+e^{\mathbf{i}\theta}|=\sqrt{(1+\cos(\theta))^2+(\sin(\theta))^2}.$ It is easy to see that $|x+y|$ is a decreasing function as $\theta$ grows from $0$ to $\pi$ where the value reaches its maximum $|x+y|=2$ when $\theta=0$ and its minimum $|x+y|=0$ when $\theta=\pi.$ This shows that if $s\leq p^{{w }-1}, |\omega_p^S|\leq s$ and equality is achieved by setting $S\subseteq T_i$ for some $i\in \{0,\cdots,p-1\},$ proving the first inequality for this special case of $s\leq p^{{w }-1}.$ Now we suppose that $s>p^{{w }-1}$ or equivalently, $s_1\geq 1.$

Note that the actual elements from each $S_i$ does not affect $\omega_p^S$ since any element from the same $S_i$ contributes $\omega_p^i$ to the sum. So we are only interested in the vector $\mathfrak{s}=(s^{(0)},\cdots, s^{(p-1)}).$ We further note that performing cyclic shift operation to $\mathfrak{s}$ does not change $|\omega_p^S|.$ For any non-negative integer $z,$ we denote by $\mathfrak{s}^{(z)}$ to be the vector we obtained by cyclic shifting $\mathfrak{s}$ to the right by $z$ position and $\mathfrak{s}^{(z)}_i$ be its $i$-th entry for $i=0,\cdots,p-1.$ So we can find a non-negative integer $z$ such that
\begin{equation*}
z\in\argmin_{z\in\{0,\cdots,p-1\}}\left\{y_z: \mathfrak{s}^{(z)}_{y_z}\neq 0, \forall i>y_z, \mathfrak{s}^{(z)}_i=0\right\}
\end{equation*}
Without loss of generality, we can assume that $z=0$ and suppose that $y_0=y\geq s_1\geq 1.$ Then $\omega_p^{S}=\sum_{i=0}^{y} s^{(i)}\omega_p^i.$ Note that by the minimality of $y_z, s^{(0)}$ must be non-zero.

\begin{claim}
There exists $S'=\bigcup_{i=0}^{s_1}S'_i\subseteq\F_{p^{w }}$ and an integer $a'$ where $S'_i=T_i$ for any $i=(a'+1)\pmod{p},\cdots, (a'+s_1-1)\pmod{p}, S'_{a'}\subseteq T_{a'}$ and $S'_{(a'+s_1)\pmod{p}}\subseteq T_{(a'+s_1)\pmod{p}}$ such that $|\omega_p^{S'}|\geq \xi.$
\end{claim}

\begin{proof}
 Note that if $y=1,$ the claim is already true by setting $S'=S.$ So assume that $y>1.$ Then there exists $j^\star\in\{0,\cdots,y-1\}$ such that $\omega_p^S$ lies between $\omega_p^{j^\star}$ and $\omega_p^{j^\star+1}.$ Then for any $i=0,\cdots, j^{\ast}-1, \omega_p^i\circ\omega_p^S \leq \omega_p^{j^\ast}\circ\omega_p^S$ and for any $i=j^{\ast}+2,\cdots, y, \omega_p^i\circ \omega_p^S\leq\omega_p^{j^\ast+1}\circ \omega_p^S.$ Here $\circ$ represents the inner product between the two complex numbers. More specifically, for any $z,z'\in \mathbb{C}, z\circ z' = |z| |z'| \cos\theta$ where $\theta$ is the angle between $z$ and $z'.$

Now suppose that there exists $x\in\{1,\cdots, y-1\}$ such that $s^{(x)}< p^{e-1}.$ Then $1\leq x\leq j^{\ast}$ or $j^{\ast}+1\leq x\leq y-1.$ We consider the case when $1\leq x\leq j^\ast$ while the proof for $j^\ast+1\leq x\leq y-1$ can be done in a similar way. As discussed before, $s^{(0)}>0.$ Hence there exists an element $\alpha$ of $S_0.$ On the other hand, since $s^{(x)}<p^{e-1},$ the set $T_x\setminus S_x$ is non-empty, suppose that $\beta\in T_x\setminus S_x.$ By the assumption above, we get that $\omega_p^S\circ \beta \geq \omega_p^S\circ \alpha.$ Note that for the case of $j^\ast+1\leq x\leq y-1,$ we choose $\alpha$ from $S_y$ and $\beta\in T_x\setminus S_x.$

Now we consider $\hat{S}=S\cup\{\beta\}\setminus\{\alpha\}.$ Then $\omega_p^{\hat{S}}=\omega_p^S+\beta-\alpha.$ Note that $|\omega_p^{\hat{S}}|^2=|\omega_p^S|^2 + |\beta-\alpha|^2 + 2|\omega_p^S||\beta-\alpha|\cos(\theta)$ where $\theta$ is the angle between $\omega_p^S$ and $\beta-\alpha.$ Recall that $\omega_P^S\circ(\beta-\alpha)\geq 0.$ Hence $\theta\in\left[-\frac{\pi}{2},\frac{\pi}{2}\right].$ So it is easy to see that $|\omega_P^{\hat{S}}|^2\geq |\omega_p^S|^2.$

The claim is proved since we can keep repeating this process until the desired form is achieved.
\end{proof}

Due to the invariance of the sum with respect to cyclic shift operation on $\mathfrak{s},$ we can assume $a'=0.$ So we can assume that $S=S_0\cup T_1\cup\cdots \cup T_{s_1-1}\cup S_{s_1}$ where $S_0\subseteq T_0$ and $S_{s_1}\subseteq T_{s_1}.$ Note that since $\omega_p^S$ is also invariant with respect to complex conjugation, without loss of generality, we can assume that $\omega_p^S\circ 1\geq \omega_p^S\circ \omega_p^{s_1}.$ So if $|S_0|<p^{e-1},$ a similar proof as the one used above can be used to prove that replacing an element $\alpha$ of $S_{s_1}$ from $S$ with an element $\beta\in T_0\setminus S_0$ does not reduce the sum. So by repeating this step until $S_0=T_0,$ we complete the proof.
\end{proof}

Next we provide a relation between product of functions over a linear code over $\F_{p^{w}}$ with the sum of products of their Fourier coefficients via the Poisson Summation formula.

\begin{lem}[Poisson Summation Formula over $\F_{p^w}$]\label{lem:psfgen}
Let $p>2$ be a prime and $w$ be a positive integer. Let $C\subseteq \F_{p^w}^n$ be a linear code with dual code $C^\bot.$ Let $f_1,\cdots, f_n:\F_{p^w}\rightarrow\mathbb{C}$ be functions. Let $\Lambda$ be defined as follows:
\begin{equation*}
\Lambda(f_1,\cdots, f_n)=\mathbb{E}_{\mathbf{x}\leftarrow C}\left[\prod_{i=1}^n f_i(x_i)\right]
\end{equation*}
where $\mathbf{x}=(x_1,\cdots, x_n).$ Then
\begin{equation*}
\Lambda(f_1,\cdots, f_n)=\sum_{{\boldsymbol \alpha}\in C^\bot}\prod_{j=1}^n \hat{f_j}(\alpha_j).
\end{equation*}
\end{lem}

The lemma can be shown using a similar proof idea as Lemma $4.16$ of~\cite{BDIR19}. For completeness, proof can be found in Supplementary Material~\ref{app:prooflempsfgen}.

\section{Leakage Resilience of Linear Codes}\label{sec:LRLinCod}
Now we are ready to proceed to the discussions of our results. In this section, we investigate the statistical distance between leakage from codewords of a fixed linear code and the leakage from a random string. We note that the analysis mainly follows the analysis in~\cite{BDIR19} with some modifications to allow for results to be applicable for a larger family of codes, i.e., linear codes over any finite fields. The objective of this section is to prove the following two theorems. Firstly, we provide a bound on the statistical distance, which can be found in Theorem~\ref{thm:new1}.
\begin{restatable}{thm}{thmnewone}\label{thm:new1}
Let $C\subseteq\F_{p^{w }}^n$ be an $[n,k,d\leq n-k+1]$ code with the dual code $C^\bot$ which is an $[n,n-k,d^\bot\leq k+1]$ code. Furthermore, let ${\boldsymbol\tau}=(\tau^{(1)},\cdots,\tau^{(n)})$ be any family of leakage functions where $\tau^{(j)}:\F_{p^{w }}\rightarrow\F_2^\mu.$ For simplicity of notation, for any set $S\subseteq \F_{p^w}^n,$ we denote by ${\boldsymbol\tau}(S)$ the random variable $(y_1,\cdots, y_n)$ where $y_i=\tau(x_i)$ and $\mathbf{x}=(x_1,\cdots, x_n)$ is uniformly sampled from $S.$

Let $c_\mu=\frac{2^\mu\sin\left(\frac{\pi}{2^\mu}\right)}{p\sin\left(\frac{\pi}{p}\right)}$ (which is less than $1$ when $2^\mu<p$). Then
\begin{equation*}
SD({\boldsymbol\tau}(C),{\boldsymbol\tau}(\mathcal{U}_n))\leq \frac{1}{2}\cdot p^{{w }(n-k)}\cdot c_\mu^{d^\bot}.
\end{equation*}
\end{restatable}

Note that the bound we have in Theorem~\ref{thm:new1} grows exponentially on $w .$  We will conclude this section by providing an improvement on Theorem~\ref{thm:new1} by eliminating the reliance of the bound on $w .$ The improved bound can be found in Theorem~\ref{thm:new2}.

\begin{restatable}{thm}{thmnewtwo}\label{thm:new2}
Let $C\subseteq \F_{p^{w }}^n$ be any $[n,k,d\leq n-k+1]$ linear code. Let ${\boldsymbol\tau}=(\tau^{(1)},\cdots, \tau^{(n)})$ be any family of leakage functions where $\tau^{(j)}:\F_{p^{w }}\rightarrow\F_2^\mu.$ Let $c_\mu^\prime = \frac{2^\mu\sin(\pi/2^\mu + \pi/2^{4\mu})}{p\sin(\pi/p)}.$ Then
\begin{equation*}
SD({\boldsymbol\tau}(C),{\boldsymbol\tau}(\mathcal{U}_n))\leq \frac{1}{2}\cdot 2^{(5\mu+1)\cdot (n-d^\bot) + \mu} \cdot (c_\mu^\prime)^{2d^\bot-n-2}
\end{equation*}
\end{restatable}

The remainder of this section is used to prove these two theorems.
\subsection{Proof of Theorem~\ref{thm:new1}}
Before we prove Theorem~\ref{thm:new1}, we first discuss some supporting lemmas that will help in our proof.
\begin{lem}\label{lem:newSDtoDualsum}
Let $C\subseteq \F_{p^{w }}^n$ be any $[n,k,d\leq n-k+1]$ linear code. Let ${\boldsymbol \tau}=(\tau^{(1)},\tau^{(2)},\cdots,\tau^{(n)})$ be any family of leakage functions where $\tau^{(j)}:\F_{p^{w }}\rightarrow\F_2^\mu.$ For any $j=1,\cdots,n,\ell_j\in \F_2^\mu$ and $x\in \F_{p^{w }},$ define $\mathds{1}_{\ell_j}:\F_{p^w}\rightarrow\{0,1\}$ where $\mathds{1}_{\ell_j}(x)=1$ if and only if $\tau^{(j)}(x)=\ell_j.$ Then
\begin{equation*}
SD({\boldsymbol\tau}(C),{\boldsymbol\tau}(\mathcal{U}_n))=\frac{1}{2}\sum_{{\boldsymbol\ell}=(\ell_1,\cdots, \ell_n)\in \F_2^{\mu\times n}}\left|\sum_{{\boldsymbol\alpha}\in C^\bot\setminus\left\{\mathbf{0}\right\}}\prod_j\widehat{\mathds{1}_{\ell_j}}(\alpha_j)\right|.
\end{equation*}
\end{lem}
Due to the page limitation, we provide the proof in Supplementary Material~\ref{app:prooflemnewsdtodualsum}.
\begin{lem}\label{lem:cmpe}
Let $\mu$ be a positive real number such that $2^\mu$ is an integer. let $c_\mu=\frac{2^\mu \sin (\pi/2^\mu)}{p\sin (\pi/p)}.$ For any sets $A_1,\cdots, A_{2^\mu}\subseteq \F_{p^{w }}$ such that $\sum_{i=1}^{2^\mu} |A_i|=p^{w },$ we have
\begin{equation*}
\left\{
\begin{array}{ccc}
\sum_{i=1}^{2^\mu}\left|\widehat{\mathds{1}_{A_i}}(\alpha)\right|\leq c_\mu, &\mathrm{~if~}\alpha\neq 0,\\
\sum_{i=1}^{2^\mu}\left|\widehat{\mathds{1}_{A_i}}(\alpha)\right|=1, &\mathrm{~if~}\alpha=0
\end{array}
\right.
\end{equation*}
where for any $A\subseteq \F_{p^{w }},\mathds{1}_A:\F_{p^{w }}\rightarrow\{0,1\}$ is the characteristic function of the set $A\subseteq\F_{p^{w }}.$ 
\end{lem}
Due to the page limitation, the proof can be found in Supplementary Material~\ref{app:prooflemcmpe}

We are now ready to prove Theorem~\ref{thm:new1}. First we restate the theorem.
\thmnewone*
\begin{proof}
By Lemma~\ref{lem:newSDtoDualsum} and triangle inequality, we have
\begin{align*}
SD({\boldsymbol\tau}(C),{\boldsymbol\tau}(\mathcal{U}_n))&=\frac{1}{2}\sum_{{\boldsymbol\ell}}\left|\sum_{{\boldsymbol\alpha}\in C^\bot\setminus\left\{\mathbf{0}\right\}}\prod_j\widehat{\mathds{1}_{\ell_j}}(\alpha_j)\right|\leq \frac{1}{2}\sum_{{\boldsymbol\ell}} \sum_{{\boldsymbol\alpha}\in C^\bot\setminus\left\{\mathbf{0}\right\}}\prod_j\left|\widehat{\mathds{1}_{\ell_j}}(\alpha_j)\right|\\
&=\frac{1}{2}\sum_{{\boldsymbol\alpha}\in C^\bot\setminus\left\{\mathbf{0}\right\}}\prod_j\left(\sum_{\ell_j}\left|\widehat{\mathds{1}_{\ell_j}}(\alpha_j)\right|\right)
\end{align*}
Note that for any $j,(\tau^{(j)})^{-1}(\ell_j)$ partitions $\F_p$ to $2^\mu$ sets. So by Lemma~\ref{lem:cmpe},
\begin{equation*}
\left\{
\begin{array}{ccc}
\sum_{\ell_j}\left|\widehat{\mathds{1}_{\ell_j}}(\alpha_j)\right|\leq c_\mu, &\mathrm{~if~}\alpha_j\neq 0,\\
\sum_{\ell_j}\left|\widehat{\mathds{1}_{\ell_j}}(\alpha_j)\right|=1, &\mathrm{~if~}\alpha_j=0
\end{array}
\right.
\end{equation*}
Note that for any ${\boldsymbol\alpha}\in C^\bot\setminus\left\{\mathbf{0}\right\}, |\{j:\alpha_j\neq 0\}|\geq d^\bot.$ Hence, since $c_\mu\leq 1,$
\begin{align*}
SD({\boldsymbol\tau}(C),{\boldsymbol\tau}(\mathcal{U}_n))&\leq\frac{1}{2}\sum_{{\boldsymbol\alpha}\in C^\bot\setminus\left\{\mathbf{0}\right\}}\prod_j\left(\sum_{\ell_j}\left|\widehat{\mathds{1}_{\ell_j}}(\alpha_j)\right|\right)= \frac{1}{2}\sum_{{\boldsymbol\alpha}\in C^\bot\setminus\left\{\mathbf{0}\right\}}c_\mu^{wt_H({\boldsymbol\alpha})}\\
&\leq \frac{1}{2} \cdot |C^\bot| \cdot c_\mu^{d^\bot}= \frac{1}{2}\cdot  p^{w (n-k)}\cdot  c_\mu^{d^\bot}.
\end{align*}
\end{proof}

\subsection{Proof of Theorem~\ref{thm:new2}}

In this subsection, we aim to prove Theorem~\ref{thm:new2} which provides an alternative direction of bounding the statistical distance from the result in Theorem~\ref{thm:new1}. Before we prove Theorem~\ref{thm:new2}, first we provide some supporting lemmas that are useful in the proof of Theorem~\ref{thm:new2}.

Firstly, instead of finding a bound for the sum of the Fourier coefficients of $\mathds{1}_{A_i}$ in each character separately as has been done in Lemma~\ref{lem:cmpe}, we consider the bound where for each $A_i,$ we consider the character that maximizes each Fourier coefficients separately.
\begin{lem}\label{lem:maxcmpe}
Let $\mu$ be some positive real such that $2^\mu$ is an integer. Let $c_\mu=\frac{2^\mu \sin(\pi/2^\mu)}{p \sin (\pi/p)}.$ For any sets $A_1,\cdots, ,A_{2^\mu}\subseteq \F_{p^{w }}$ such that $\sum_{i=1}^{2^\mu} |A_i|=p^{w },$ we have
\begin{equation*}
\sum_{i=1}^{2^\mu} \max_{\alpha\neq 0} |\widehat{\mathds{1}_{A_i}}(\alpha)|\leq c_\mu.
\end{equation*}
\end{lem}

Proof can be found in Supplementary Material~\ref{app:prooflemmaxcmpe}. Next we consider an improvement on the bound for the statistical distance.

\begin{lem}\label{lem:newSDboundpe}
Let $C$ be an $[n,k,d\leq n-k+1]$ linear code with parity check matrix $H=(\mathbf{h}_1^T | \mathbf{h}_2^T | \cdots | \mathbf{h}_n^{T})$ where $\mathbf{h}_i^T$ is a column vector of length $n-k,$ which is the $i$-th column of $H.$ We also let $C^\bot$ be its dual with parameter $[n,n-k,d^\bot\leq k+1].$ Partition the indices of the columns of $H$ into three disjoint subsets $I_1,I_2$ and $I_3$ where $|I_1|=|I_2|=n-d^\bot+1.$ Let $\mu$ be a positive integer and ${\boldsymbol\tau}=(\tau^{(1)},\cdots, \tau^{(n)})$ be any family of leakage functions where $\tau^{(j)}:\F_{p^{w }}\rightarrow\F_2^\mu.$ We define $\mathds{1}_{\ell_j}(x)$ as before. Then

\begin{equation*}
SD({\boldsymbol\tau}(C),{\boldsymbol\tau}(\mathcal{U}_n))\leq \frac{1}{2}\cdot 2^{\mu(n-d^\bot+1)}\cdot\sum_{\{\ell_j\}_{j\in I_3}}\max_{{\boldsymbol\beta}\in \F_{p^{w }}^{n-k}\setminus\left\{\mathbf{0}\right\}}\prod_{j\in I_3}\left|\widehat{\mathds{1}_{\ell_j}}(\langle{\boldsymbol\beta},{\boldsymbol h}_j\rangle)\right|
\end{equation*}
\end{lem}

Due to the page limitation, the proof can be found in Supplementary Material~\ref{app:prooflemnewsdboundpe}.

Lemma~\ref{lem:newSDboundpe} bounds the statistical distance of a linear code using some functions that are related to a submatrix of its parity check matrix. Note that such submatrix defines a linear code with length $|I_3|.$ Lastly, we provide a bound for such code in Lemma~\ref{lem:newcmgenpe}.

\begin{lem}\label{lem:newcmgenpe}
Let $D\subseteq\F_{p^{w }}^\kappa$ be any code of distance at least $d.$ Consider an arbitrary family of leakage functions ${\boldsymbol\tau}=\left(\tau^{(1)},\cdots, \tau^{(\kappa)}\right)$ where $\tau^{(j)}:\F_{p^{w }}\rightarrow\F_2^\mu.$ Recall that we defined $\mathds{1}_{\ell_j}(x)=1$ if $\tau^{(j)}(x)=\ell_j$ and $0$ otherwise. Let $c_\mu'=\frac{2^\mu\sin(\pi/2^\mu+\pi/2^{4\mu})}{p\sin(\pi/p)}.$ Then
\begin{equation*}
\sum_{{\boldsymbol\ell}\in \F_2^{\mu\times \kappa}} \max_{{\boldsymbol\alpha}\in D\setminus\{\mathbf{0}\}}\prod_{j=1}^\kappa \left|\widehat{\mathds{1}_{\ell_j}}(\alpha_j)\right|\leq 2^{(4\mu+1)\cdot(\kappa-d)}\cdot (c_\mu')^\kappa.
\end{equation*}
\end{lem}
Due to page limitation, the proof can be found in Supplementary Material~\ref{app:prooflemnewcmgenpe}.

Now we are ready to prove Theorem~\ref{thm:new2}. First we restate the theorem.
\thmnewtwo*
\begin{proof}
By Lemma~\ref{lem:newSDboundpe}, we have
\begin{equation*}
SD({\boldsymbol\tau}(C),{\boldsymbol\tau}(\mathcal{U}_n))\leq \frac{1}{2}\cdot 2^{\mu(n-d^\bot+1)}\cdot\sum_{\{\ell_j\}_{j\in I_3}}\max_{{\boldsymbol\beta}\in \F_{p^{w }}^{n-k}\setminus\left\{\mathbf{0}\right\}}\prod_{j\in I_3}\left|\widehat{\mathds{1}_{\ell_j}}(\langle{\boldsymbol\beta},\mathbf{h}_j\rangle)\right|
\end{equation*}
where $I_3\subseteq[n]$ of size $\kappa\triangleq n-2(n-d^\bot+1)=2d^\bot-n-2.$ Let $D=\{(x_j)_{j\in I_3}: \mathbf{x}\in C^\bot\}.$ Since $C^\bot$ is an $[n,n-k,d^\bot\leq k+1]$ code, $D$ is a $[\kappa,k',d']$ code where $k'\leq n-k$ and $d'\geq d^\bot-(n-\kappa).$ This implies that $\kappa-d'\leq n-d^\bot.$ Then by Lemma~\ref{lem:newcmgenpe},

\begin{align*}
\sum_{\{\ell_j\}_{j\in I_3}}\max_{{\boldsymbol\beta}\in \F_{p^{w }}^{n-k}\setminus\left\{\mathbf{0}\right\}}\prod_{j\in I_3}\left|\widehat{\mathds{1}_{\ell_j}}(\langle{\boldsymbol\beta},\mathbf{h}_j\rangle)\right|&= \sum_{\{\ell_j\}_{j\in I_3}} \max_{{\boldsymbol\alpha}\in D\setminus\{\mathbf{0}\}}\prod_{j\in I_3}\left|\widehat{\mathds{1}_{\ell_j}}(\alpha_j)\right|\\
&\leq 2^{(4\mu+1)\cdot(\kappa-d')}\cdot (c_\mu^\prime)^\kappa\\
&\leq 2^{(4\mu+1)\cdot(n-d^\bot)}\cdot (c_\mu^\prime)^{2d^\bot-n-2}.
\end{align*}
So we have

\begin{align*}
SD({\boldsymbol\tau}(C),{\boldsymbol\tau}(\mathcal{U}_n))&\leq \frac{1}{2}\cdot 2^{\mu(n-d^\bot+1)}\cdot 2^{(4\mu+1)\cdot(n-d^\bot)}\cdot (c_\mu^\prime)^{2d^\bot-n-2}\\
&= \frac{1}{2}\cdot 2^{(5\mu+1)\cdot (n-d^\bot) + \mu} \cdot (c_\mu^\prime)^{2d^\bot-n-2}
\end{align*}

which completes the proof.

\end{proof}

\section{Local Leakage Resilience of Additive Secret Sharing Schemes and Shamir's Secret Sharing Schemes over Arbitrary Finite Fields}\label{sec:SSSgenFF}

In this section, we apply our analysis in Section~\ref{sec:LRLinCod} to the family of MDS codes. This generalizes the result of~\cite{BDIR19} to threshold secret sharing schemes defined over field extensions. In order to achieve this, we state the results in Theorems~\ref{thm:new1} and~\ref{thm:new2} when applied to MDS codes.
%we consider a special case for the analysis we have discussed in Section~\ref{sec:LRLinCod} when the code is MDS. Such results provide generalization for the results on MDS codes, additive secret sharing schemes and Shamir's secret sharing schemes discussed in~\cite{BDIR19} to the case when the underlying field is not a prime field. First, we restate Theorems~\ref{thm:new1} and~\ref{thm:new2} when the linear code is MDS.

\begin{cor}\label{cor:MDSthm12}
Let $C\subseteq\F_{p^{w }}^n$ be an MDS $[n,k,n-k+1]$ code with the dual code $C^\bot$ which is an $[n,n-k,k+1]$ code. Furthermore, let ${\boldsymbol\tau}=(\tau^{(1)},\cdots,\tau^{(n)})$ be any family of leakage functions where $\tau^{(j)}:\F_{p^{w }}\rightarrow\F_2^\mu$ for some integer $\mu<\log p.$ Let $c_\mu=\frac{2^\mu\sin\left(\frac{\pi}{2^\mu}\right)}{p\sin\left(\frac{\pi}{p}\right)}$ and $c_\mu^\prime = \frac{2^\mu\sin(\pi/2^\mu + \pi/2^{4\mu})}{p\sin(\pi/p)}.$ Then
\begin{equation*}
SD({\boldsymbol\tau}(C),{\boldsymbol\tau}(\mathcal{U}_n))\leq \frac{1}{2}\cdot \min\left(p^{{w }(n-k)}\cdot c_\mu^{k+1},2^{(5\mu+1)\cdot (n-k-1) + \mu} \cdot (c_\mu^\prime)^{2k-n}\right).
\end{equation*}
%and
%\begin{equation*}
%SD({\boldsymbol\tau}(C),{\boldsymbol\tau}(\mathcal{U}_n))\leq \frac{1}{2}\cdot 2^{(5\mu+1)\cdot (n-k-1) + \mu} \cdot (c_\mu^\prime)^{2k-n}.
%\end{equation*}
\end{cor}

Using the same approach in~\cite[Theorem $4.7$]{BDIR19}, we have the following result on additive secret sharing over $\F_{p^w}.$
\begin{cor}\label{cor:gen47}
Let $C\subseteq\F_{p^w}^n$ be the code generated with codewords having entries being valid additive shares of $0.$ Letting ${\boldsymbol \tau}=(\tau^{(1)},\cdots,\tau^{(n)})$ be any family of leakage functions where each $\tau^{(j)}$ has $\mu$ bit output for some $\mu<\log p.$ Letting $c_\mu=\frac{2^\mu \sin(\pi/2^\mu)}{p\sin(\pi/p)}.$ Then
\begin{equation*}
SD({\boldsymbol \tau}(C),{\boldsymbol \tau}(U_n))\leq \frac{1}{2}\cdot 2^\mu\cdot c_\mu^{n-2}.
\end{equation*}
\end{cor}

Utilizing Corollaries~\ref{cor:MDSthm12} and~\ref{cor:gen47} we can obtain similar results on leakage resilience of both additive secret sharing schemes and Shamir's secret sharing schemes as discussed in~\cite[Section 4.2.2 and Section 4.2.3]{BDIR19}.

\section{Local Leakage Resilience of Algebraic Geometric codes based Ramp Secret Sharing Scheme}~\label{sec:AGSSS}
Let $q=p^w$ and $F/\F_q$ be a function field of genus $\mathfrak{g}$ and at least $n+2$ distinct $\F_q$ rational places $P_\infty,P_0,\cdots, P_n.$ Set $\mathcal{P}=\{P_1,\cdots, P_n\}$ and $\mathcal{P}_0=\{P_0\}\cup \mathcal{P}.$ Consider a ramp secret sharing scheme $AGSh_{mP_\infty,\mathcal{P}}$ over $\F_{q}$ based on $C(mP_\infty,\mathcal{P})$ constructed using the technique discussed in Section~\ref{subsec:AGCode}. In this section, we utilize Theorems~\ref{thm:new1} and~\ref{thm:new2} to establish the local leakage resilience of $AGSh_{mP_\infty,\mathcal{P}}.$

\subsection{Local Leakage Resilience of $AGSh_{mP_\infty,\mathcal{P}}$}

First we recall the the local leakage scenario we are considering. We use the algebraic-geometric code based ramp secret sharing scheme $AGSh_{mP_\infty,\mathcal{P}}$ with $t=m-2\mathfrak{g}$ privacy and $2\mathfrak{g}+t+1$ reconstruction to secret share a secret to a set of players $\mathtt{U}=\{U_1,\cdots, U_n\}.$ A passive adversary then chooses $\Theta\subseteq\mathtt{U}$ of size $|\Theta|=\theta<t$ to control as well as the leakage function ${\boldsymbol \tau}=(\tau^{(1)},\cdots, \tau^{(n)})$ with each leakage function outputting $\mu<\log p$ bits each. This defines the view of the adversary as $Leak_{\Theta,{\boldsymbol \tau}}.$ Note that in the view of the adversary, due to the $\theta$ shares that he has learned in full, the remaining required information is less than what is originally needed. In fact, we can no longer apply Theorems~\ref{thm:new1} and~\ref{thm:new2} directly using $C(mP_\infty,\mathcal{P})$ as the linear code. The following lemma provides a linear code $\mC$ that is equivalent to the view of the adversary after learning the shares of the players in $\Theta.$

\begin{lem}\label{lem:compsh}
Consider the secret sharing scheme $AGSh_{n,r,t}=AGSh_{mP_\infty,\mathcal{P}}$ which is used to secret share $s\in \F_q$ where each user $U_i$ gets a share $s_i\in \F_q.$ Let $\Theta\subseteq\mathtt{U}$ be a set of $\theta\leq t=m-2\mathfrak{g}$ players that is corrupted by the adversary and set $\bar{\Theta}=\mathtt{U}\setminus\Theta.$ Consider the following experiment. Let the values of the shares held by the corrupted parties be $\mathbf{x}^{(\Theta)}.$ Let $\left.AGSh_{n,r,t}(s)\right|_{\mathbf{s}^{(\Theta)}=\mathbf{x}^{(\Theta)}}$ be the distribution on the shares conditioned on the revealed values $\mathbf{s}^{(\Theta)}$ being $\mathbf{x}^{(\Theta)}.$ Then there exists an $[n-\theta,m-\theta-\mathfrak{g},d'\geq n-m+1]$ code $\mathcal{C}'\subseteq \F_{q}^{n-\theta}$ with dual code $(\mathcal{C}')^\bot$ with parameter $[n-\theta,n-m+\mathfrak{g},(d')^\bot\geq m-\theta-2\mathfrak{g}+1]$ and a shift vector $\mathbf{b}\in \F_q^n$ such that

\begin{equation*}
\left.AGSh_{n,r,t}(s)\right|_{\mathbf{s}^{(\Theta)}=\mathbf{x}^{(\Theta)}}\equiv \left\{\left(\mathbf{y}^{\bar{(\Theta)}}|\mathbf{0}^{(\Theta)}\right)+\mathbf{b}:\mathbf{y}^{\bar{(\Theta)}}\leftarrow \mathcal{C}'\right\}.
\end{equation*}
\end{lem}

\begin{proof}
Without loss of generality, assume that $\Theta=\{U_1,\cdots, U_\theta\}.$ Then by Lemma~\ref{lem:surj}, there exists $p\in\mathcal{L}(mP_\infty)$ such that $p(P_i)=x_i$ for any $i\in \Theta.$ So $\left.AGSh_{n,r,t}(s)\right|_{\mathbf{s}^{(\Theta)}=\mathbf{x}^{(\Theta)}}\equiv \left.AGSh_{n,r,t}(0)\right|_{\mathbf{s}^{(\Theta)}=\mathbf{0}}+\mathbf{p}-(p(P_0)-s)\mathbf{1}_n$
where $\mathbf{1}_n$ is the all-one vector of length $n.$ Set $\mathbf{b}=\mathbf{p}-(p(P_0)-s)\mathbf{1}_n.$ Now we consider $\left.AGSh_{n,r,t}(0)\right|_{\mathbf{s}^{(\Theta)}=\mathbf{0}}.$

Note that $\left.AGSh_{n,r,t}(0)\right|_{\mathbf{s}^{(\Theta)}=\mathbf{0}}$ can be seen as follows: Sample $f\in \mathcal{L}(mP_\infty-P_0-P_1-\cdots-P_\theta)$ and $\mathcal{C}'$ is defined to be $\{(y^{\bar{(\theta)}}: y=(f(P_1),\cdots, f(P_n)), f\in \mathcal{L}(mP_\infty-P_0-\sum_{i\in \Theta}P_i)\}.$ Then $\left.AGSh_{n,r,t}(0)\right|_{\mathbf{s}^{(\Theta)}=\mathbf{0}}$ is equivalent to $\mathcal{C}',$ which is an $[n-\theta,m-\theta-\mathfrak{g},d'\geq n-m+1]$ code with dual code $(\mathcal{C}')^\bot$ an $[n-\theta,n-m+\mathfrak{g},(d')^\bot\geq m-\theta-2\mathfrak{g}+1]$ code.
\end{proof}

Using this observation, we can then provide the leakage resilience of the ramp secret sharing scheme $AGSh_{n,r,t}(s)$ which is defined over $\F_{p^{w }}.$
\begin{cor}\label{cor:AGShOverpe}
The ramp secret sharing scheme $AGSh_{n,r,t}(s)$ defined over $\F_{p^{w }}$  is $(\theta,\mu,\epsilon)$-LL resilient where
\begin{equation*}
\epsilon=\min\left(p^{{w }\left(n-t-\frac{r-1-t}{2}\right)}\cdot c_\mu^{t-\theta+1},2^{(n-t-1)(5\mu+1)+\mu}\cdot (c_\mu^\prime)^{2t-n-\theta}\right).
\end{equation*}
\end{cor}

\begin{proof}
Let $s,s'\in \F_{p^w}.$ For simplicity of notation, we denote by $SD$ the statistical distance between the outputs of ${\boldsymbol \tau}$ with inputs being the shares generated by $s$ and $s'$ respectively under the assumption that $\mathbf{s}^{(\Theta)}=\mathbf{x}^{(\Theta)}.$ By Lemma~\ref{lem:compsh}, having the adversary seeing the values of $\mathbf{s}^{(\Theta)}=\mathbf{x}^{(\Theta)},$ the adversary specifies any family of $\mu$-bit output leakage functions ${\boldsymbol\tau}^{\bar{(\Theta)}}=(\tau^{(i)})_{i\in \bar{\Theta}}.$ Since the distribution of $AGSh_{n,r,t}(s)$ after the leak of $\mathbf{s}^{(\Theta)}$ is equivalent to an $[n-\theta,m-\theta-\mathfrak{g},d'\geq n-m+1]$ code, by Theorem~\ref{thm:new1}, $SD\leq \frac{1}{2}\cdot p^{{w }(n-m+\mathfrak{g})}\cdot c_\mu^{(d')^\bot}\leq \frac{1}{2}\cdot p^{w (n-m+\mathfrak{g})}\cdot c_\mu^{m-\theta-2\mathfrak{g}+1}= \frac{1}{2}\cdot p^{w \left(n-t-\frac{r-1-t}{2}\right)}\cdot c_\mu^{t-\theta+1}.$ Using triangle inequality, we have that for any $s\neq s'\in \F_{p^{w }},$ we have

\begin{equation*}
SD\leq  p^{w \left(n-t-\frac{r-1-t}{2}\right)}\cdot c_\mu^{t-\theta+1}.
\end{equation*}

\noindent Using the same argument based on Theorem~\ref{thm:new2}, we have
\begin{equation*}
SD\leq  2^{(n-t-1)(5\mu+1)+\mu}\cdot (c_\mu^\prime)^{2t-n-\theta}
\end{equation*}

\end{proof}

\noindent In order to obtain a ramp secret sharing scheme that can be constructed in polynomial time, we construct the code by applying the Garcia-Stichtenoth tower~\cite{GS96} to construct the AG code. However, for such method to be applicable, we need to be working over $\F_q$ where $q$ is a square of a prime. In other words, our resulting ramp secret sharing scheme is defined over $\F_q$ where $q=p^2.$ Hence, applying this result when $w=2,$ we have our first main result.
\mtone*

\subsection{Construction of Algebraic-Geometric Codes-based Ramp Secret Sharing Scheme by Concatenation Scheme}
Recall that for any element in $\F_{p^{w }}$ where $w $ is a positive integer that can be factorized to $w =uv$ for some positive integers $u$ and $v,$ there is a vector space isomorphism between $\F_{p^{w }}$ with $\F_{p^{u}}^{v}.$ Fix one of such isomorphisms and denote it by $\Pi_{u,v}.$

Consider an AG-codes based ramp secret sharing scheme defined over $\F_{p^{w }}.$ Then using $\Pi_{u,v},$ we can map the $n$ shares $s_1,\cdots, s_n\in \F_{p^{w }}$ to $vn$ shares $s'_1,\cdots, s'_{vn}\in \F_{p^{u}}$ where for any
$i=1,\cdots, n, (s'_{v(i-1)+1},\cdots, s'_{vi})=\Pi_{u,v}(s_i).$ Then the resulting ramp secret sharing schemes has $N=vn$ players providing $T=m-2\mathfrak{g}$ privacy and $R=(v-1)n+m+1=\frac{v-1}{v}N+m+1$ reconstruction. Note that this can be easily verified by noting that learning any $T$ shares provides at most $T$ shares of the original AG-code based ramp secret sharing scheme that is defined over $\F_{p^w}.$ Hence by definition, the adversary learns no information about the original secret from any of such $T$ shares. On the other hand, by Pigeon Hole principle, having $R$ shares, there are at least $m+1$ of $i\in \{1,\cdots, n\}$ such that we learned $(s'_{v(i-1)+1},\cdots, s'_{vi}).$ Hence, using the isomorphism $\Pi_{u,v},$ we can learn at least $m+1$ of the original shares $s_i\in \F_{p^w}.$ By definition, such information is sufficient to reconstruct the original secret. Denote the resulting ramp secret sharing scheme by $EAGSh_{N,R,T}$ that secretly shares a secret $s\in \F_{p^{w }}$ to $N=v n$ players over $\F_{p^{u}}.$ Here we denote the set of $N$ players by $\hat{\mathtt{U}}=\{U_1,\cdots, U_N\}.$

 We consider the extension of Lemma~\ref{lem:compsh}.

\begin{lem}\label{lem:newcompsh}
Let $\Theta\subseteq\hat{\mathtt{U}}$ be a set of $\theta\leq T$ players. Consider the following experiment where for a given secret $s\in \F_{p^{w}},$ the $N$ shares $\mathbf{s}=(s'_1,\cdots, s'_N)=EAGSh_{N,R,T}(s)$ are generated while the shares $s_i$ of $U_i$ for all $U_i\in \Theta$ are leaked. Let these values be  $\mathbf{x}^{(\Theta)}.$ Let $EAGSh_{N,R,T}(s)|_{\mathbf{s}^{(\Theta)}=\mathbf{x}^{(\Theta)}}$ be the distribution on the shares conditioned on the revealed values $\mathbf{s}^{(\Theta)}$ being $\mathbf{x}^{(\Theta)}.$ Then there exists an $\left[N-v\theta, v(m-\theta-\mathfrak{g}),D'\geq \frac{N}{v}-m+1\right]$ code $\mathcal{C}'\subseteq \F_{p^{u}}^{N-v\theta}$ with dual code $(\mathcal{C}')^\bot$ with parameter $[N-v\theta, N-vm+v\mathfrak{g},(D')^\bot\geq m-\theta-2\mathfrak{g}+1]$ and a shift vector $\mathbf{b}\in \F_{p^{u}}^N$ such that
\begin{equation*}
EAGSh_{N,R,T}(s)|_{\mathbf{s}^{(\Theta)}=\mathbf{x}^{(\Theta)}}\equiv \left\{\left.\left(\mathbf{y}^{\bar{(\Theta)}}\right|\mathbf{0}^{(\Theta)}\right)+\mathbf{b}:\mathbf{y}^{\bar{(\Theta)}}\leftarrow \mathcal{C}'\right\}.
\end{equation*}
\end{lem}
\begin{proof}
Suppose that $\Theta\subseteq[N]$ with $|\Theta|=\theta<t.$ Then in the worst case, the leaks reveal the $\theta$ values in the corresponding ramp secret sharing scheme over $\F_{p^{w }}.$ Suppose that the $\theta$ values leaked in the corresponding ramp secret sharing scheme over $\F_{p^{w }}$ is $s_1,\cdots, s_\theta.$ Then by Lemma~\ref{lem:compsh}, there exists an $[n-\theta,m-\theta-\mathfrak{g},\hat{d}\geq n-m+1]$ code $\hat{\mathcal{C}}\subseteq \F_{p^{w }}^{n-\theta}$ with dual code $\hat{\mathcal{C}}^\bot$ with parameter $[n-\theta,n-m+\mathfrak{g},\hat{d}^\bot\geq m-\theta-2\mathfrak{g}+1]$ and a shift vector $\hat{\mathbf{b}}\in \F_{p^{w }}^n$ such that the distribution on the corresponding $n$ shares over $\F_{p^{w }}$ conditioned on the revealed values is equivalent to $\left\{\left(\hat{\mathbf{y}}^{(\overline{\Theta'})}|\mathbf{0}^{(\Theta')}\right)+\hat{\mathbf{b}}:\hat{\mathbf{y}}^{(\overline{\Theta'})}\leftarrow \hat{\mathcal{C}}\right\}.$ Then, using the isomorphism $\Pi_{u,v},$ setting $\mathcal{C}'=\Pi_{u,v}(\hat{\mathcal{C}})\subseteq \F_{p^{u}}^{N-v\theta}$ with parameter $[N-v\theta, v(m-\theta-\mathfrak{g}),D'\geq \frac{N}{v}-m+1]$ and dual $(\mathcal{C}')^\bot=\Pi_{u,v}(\hat{\mathcal{C}}^\bot)\subseteq \F_{p^{u}}^{N-v\theta}$ with parameter $[N-v\theta, N-vm+v\mathfrak{g},(D')^\bot\geq m-\theta-2\mathfrak{g}+1]$ along with a shift vector $\mathbf{b}=\Pi_{u,v}(\hat{\mathbf{b}})\in \F_{p^{u}}^N,$ we have the desired result.
\end{proof}

Then we have the following result on the leakage resilience of ramp secret sharing schemes defined in this way.
\begin{cor}\label{cor:EAGShOverpe}
The ramp secret sharing scheme $EAGSh_{N,R,T}(s)$ defined over $\F_{p^{u}}$ is $(\theta,\mu,\epsilon)$- LL resilient where
\begin{equation*}
\epsilon=p^{\frac{w }{2}\left(\frac{v+1}{v}N-T-R+1\right)}\cdot c_\mu^{T-\theta+1}~\mathrm{or}~~\epsilon=2^{(N-(v-1)\theta-T-1)(5\mu+1)+\mu}\cdot (c_\mu^\prime)^{2T+(v-2)\theta-N}
\end{equation*}
\end{cor}
\begin{proof}
The proof uses the same argument as Corollary~\ref{cor:AGShOverpe}. For simplicity of notation, we denote $\mathcal{S}$ the statistical distance between ${\boldsymbol\tau}\left(EAGSh_{N,R,T}(s)|_{\mathbf{s}^{(\Theta)}=\mathbf{x}^{(\Theta)}}\right)$ and ${\boldsymbol\tau}\left(EAGSh_{N,R,T}(s')|_{\mathbf{s}^{(\Theta)}=\mathbf{x}^{(\Theta)}}\right)).$ Then for any $s\neq s'\in \F_{p^{w }},$ Theorem~\ref{thm:new1} implies
\begin{equation*}
\mathcal{S}\leq  p^{u\left(\frac{(v+1)N}{2}-\frac{v}{2}T-\frac{v}{2}R+\frac{v}{2}\right)}\cdot c_\mu^{T-\theta+1}=p^{\frac{w}{2}\left(\frac{v+1}{v}N-T-R+1\right)}\cdot c_\mu^{T-\theta+1}.
\end{equation*}
Similarly, Theorem~\ref{thm:new2} implies
\begin{equation*}
\mathcal{S}\leq  2^{(N-(v-1)\theta-T-1)(5\mu+1)+\mu}\cdot (c_\mu^\prime)^{2T+(v-2)\theta-N}
\end{equation*}
\end{proof}

As before, we aim to obtain a ramp secret sharing scheme that can be constructed in polynomial time. So by applying Garcia-Stichtenoth tower~\cite{GS96}, the original AG code is defined over $\F_{p^2}.$ We can then apply this result to the special case when $w=v=2$ which is discussed in the following main theorem.

\mttwo*

Lastly, we provide some analysis on the results presented in Main Theorems~\ref{mt:1} and~\ref{mt:2} to determine whether it is possible to have a ramp secret sharing scheme defined over a field extension that has a better leakage resilience property compared to a ramp secret sharing scheme defined over a prime field with comparable parameters. In other words, we assume that they are defined in fields with approximately the same size, approximately the same length $N$ and privacy $T.$ Having fixed these parameters, we can then determine the reconstruction levels, which is denoted by $R_1=T+\frac{N}{\sqrt{q}-1}+1$ for construction in Main Theorem~\ref{mt:1} and $R_2=T+\frac{N}{2}+\frac{N}{q-1}+T+1$ for construction in Main Theorem~\ref{mt:2}. Then, assuming that $\mu<\log\sqrt{q},$ we can find some values of the parameters such that the leakage resilience property of the construction in Main Theorem~\ref{mt:1} is better than that of the construction in Main Theorem~\ref{mt:2}. The result is summarized in Lemma~\ref{lem:mt1better}.

\begin{lem}\label{lem:mt1better}
Let $N,T,R_1,R_2,\theta, \mu,q$ be as defined above such that $T<\frac{\sqrt{q}}{q-1}N-1$ while we set $\theta\geq \max\left(2T-\frac{\sqrt{q}}{q-1}N+1,(\rho+1)T-\rho N-(\rho-1)\right)$ where $\rho= \frac{\log\left(\frac{5}{3}\right)}{\log\left(\frac{51}{50}\right)}.$ Furthermore, let $\mu$ be chosen such that $\max\left(2,\frac{1}{5}\log(\sqrt{q})\right)\leq \mu<\log(\sqrt{q})$ for some $q> 16.$ Assume that an adversary $\mathcal{A}$ has access to the full shares of any $\theta$ players and $\mu$ bits from the remaining players. Let $AGSh_{n,R_1,T}$ be a ramp secret sharing scheme defined over $\F_{q_1}$ with $N$ players, $T$ privacy by Main Theorem~\ref{mt:1} for some prime square $q_1\approx q$ and $EAGSh_{n,R_2,T}$ be a ramp secret sharing scheme defined over $\F_{q_2}$ with $N$ players and $T$ privacy by Main Theorem~\ref{mt:2} for some prime $q_2\approx q.$ Then when $N$ is sufficiently large, $AGSh_{N,R_1,T}$ has a stronger leakage resilience against $\mathcal{A}$ compared to $EAGSh_{N,R_2,T}.$
\end{lem}
\subsection{Proof of Lemma~\ref{lem:mt1better}}
In this section, we prove Lemma~\ref{lem:mt1better}. In order for the comparison to be done, we assume that they have approximately the same size of $q$ (where $q$ is a square of a prime for Main Theorem~\ref{mt:1} and it is a prime number for Main Theorem~\ref{mt:2}). Since we want $\mu\geq 1,$ we assume that the $\sqrt{q}\geq 3.$ Furthermore, we will also consider the ramp secret sharing schemes to have approximately the same length $N$ and privacy $T.$ Let $\mathfrak{g}_1$ and $\mathfrak{g}_2$ be the genus of the function field considered in the two main theorems respectively. Then we have $\mathfrak{g}_1\approx \frac{N}{\sqrt{q}-1}$ and $\mathfrak{g}_2\approx\frac{N}{q-1}$~\cite{GS96}. Then, denoting the reconstruction levels of the two main theorems to be $R_1$ and $R_2$ respectively, it can be easily verified that $R_1=T+\frac{N}{\sqrt{q}-1}+1<T+\frac{N}{2}+\frac{N}{q-1}+T+1=R_2$ if $q\geq 6.$ Lastly, we will also consider that the number of corrupted players $\theta$ and the number of leaked bits from the remaining players $\mu$ in both cases are set to be the same. Hence we assume that $\mu<\log \sqrt{q}.$

We aim to show that in some situation, the leakage resilience of $AGSh_{N,R_1,T}$ is strictly stronger than $EAGSh_{N,R_2,T}.$ More specifically, in some cases, the upper bound of statistical distance between the leak from any two secrets shared by $AGSh_{N,R_1,T}$ is strictly smaller than the statistical distance of those shared by $EAGSh_{N,R_2,T}.$ For simplicity of notation, we set $c_1=\frac{2^\mu \sin\left(\frac{\pi}{2^\mu}\right)}{\sqrt{q}\sin\left(\frac{\pi}{\sqrt{q}}\right)}, c_1^\prime=\frac{2^\mu \sin\left(\frac{\pi}{2^\mu}+\frac{\pi}{2^{4\mu}}\right)}{\sqrt{q}\sin\left(\frac{\pi}{\sqrt{q}}\right)}, c_2=\frac{2^\mu \sin\left(\frac{\pi}{2^\mu}\right)}{q\sin\left(\frac{\pi}{q}\right)}$ and $c_2^\prime=\frac{2^\mu \sin\left(\frac{\pi}{2^\mu}+\frac{\pi}{2^{4\mu}}\right)}{q\sin\left(\frac{\pi}{q}\right)}.$ Furthermore, we also simplify the notations for the $\epsilon$ guarantee for the main theorems. More specifically, we let $\epsilon_1=q^{\left(N-T-\mathfrak{g}_1\right)}\cdot c_1^{T-\theta+1},~\epsilon_2=2^{(N-T-1)(5\mu+1)+\mu}\cdot (c_1^\prime)^{2T-N-\theta},\epsilon_3=q^{\left(N-2T-\mathfrak{g}_2\right)}\cdot c_2^{T-\theta+1}$ and $\epsilon_4=2^{(N-\theta-T-1)(5\mu+1)+\mu}\cdot (c_2^\prime)^{2T-N}.$

First we compare $\epsilon_1$ and $\epsilon_2.$ Note that in contrast to~\cite{BDIR19} where $q>N$ and hence $\epsilon_1$ is asymptotically larger than $\epsilon_2,$ in our case, we can let $q$ to be a constant even when $N$ goes to infinity. Lemma~\ref{pro:1vs2} shows that such difference enables $\epsilon_1$ to be smaller than $\epsilon_2.$ Combined with Main Theorem~\ref{mt:1}, Lemma~\ref{pro:1vs2} implies that $AGSh_{n,R_1,T}$ is $(\theta, \mu, \epsilon)$-LL Resilient for $\theta<T,\mu<\log \sqrt{q}$ and $\epsilon=\min(\epsilon_1,\epsilon_2)=\epsilon_1.$

\begin{prop}\label{pro:1vs2}
Let $\max(2,\frac{2}{5} \log(\sqrt{q}))\leq \mu <\log \sqrt{q}$ and assume that we have $N-T+1\geq \frac{\log\left(\frac{51}{50}\right)}{\log\left(\frac{5}{3}\right)}\left(T-\theta+1\right).$ Then when $N$ is sufficiently large, we have $\epsilon_1<\epsilon_2.$
\end{prop}
\begin{proof}
Let $\Delta=\frac{\epsilon_1}{\epsilon_2}.$ We aim to show that $\Delta<1.$ By a simple algebraic manipulation, we have
\begin{eqnarray*}
\Delta\leq \left(\frac{\sin\left(\frac{\pi}{2^\mu}+\frac{\pi}{2^{4\mu}}\right)}{\sin\left(\frac{\pi}{2^\mu}\right)}\right)^{T-\theta+1}\cdot \left(\frac{c_1'}{2}\right)^{N-T+1}\cdot 2^{(7-\mathfrak{g}_1)\mu+2}.
\end{eqnarray*}
Assuming that $\mu\geq 2,$ we can verify that $\frac{\sin\left(\frac{\pi}{2^\mu}+\frac{\pi}{2^{4\mu}}\right)}{\sin\left(\frac{\pi}{2^\mu}\right)}< \frac{51}{50}.$ By~\cite[Proposition~$A.1$]{BDIR19}, $c_1^\prime\leq 2^{-\frac{1}{2^{2\mu+2}}+\frac{4}{q}}.$ Note that since we want $\mu\geq 2,$ we have $q>16.$ Utilizing such observation, we have $\frac{c_1'}{2}<2^{-\frac{49}{64}}<\frac{3}{5}.$ Together with the assumption on the values of $\theta, T$ and $N,$ these upper bounds ensure that  $\left(\frac{\sin\left(\frac{\pi}{2^\mu}+\frac{\pi}{2^{4\mu}}\right)}{\sin\left(\frac{\pi}{2^\mu}\right)}\right)^{T-\theta+1}\cdot \left(\frac{c_1'}{2}\right)^{N-T+1}<1.$ Recall that $\mathfrak{g}_1\approx\frac{N}{\sqrt{q}-1}.$ Hence when $N$ is sufficiently large, we again have $2^{(7-\mathfrak{g}_1)\mu+2}\leq 1.$ This completes the proof that $\Delta<1.$

%\begin{claim}
%When $\mu\geq 2,$ we have $\frac{\sin\left(\frac{\pi}{2^\mu}+\frac{\pi}{2^{4\mu}}\right)}{\sin\left(\frac{\pi}{2^\mu}\right)}\leq \frac{51}{50}.$
%\end{claim}
%\begin{proof}
%Note that since $\mu\geq 2,$ we have $\frac{\pi}{2^{4\mu}}\leq \frac{1}{64}\frac{\pi}{2^\mu}.$ So we have $\sin\left(\frac{\pi}{2^\mu}+\frac{\pi}{2^{4\mu}}\right)\leq \sin\left(\frac{65}{64} \frac{\pi}{2^\mu}\right).$ Then a simple algebraic manipulation gives us $\frac{\sin\left(\frac{\pi}{2^\mu}+\frac{\pi}{2^{4\mu}}\right)}{\sin\left(\frac{\pi}{2^\mu}\right)}-1\leq \sin\left(\frac{\pi}{256}\right).$ It can then be verified that $\sin\left(\frac{\pi}{256}\right)\leq \frac{1}{50}$ which completes the proof.
%\end{proof}

\end{proof}

Similarly, we establish some situations where $\epsilon_3<\epsilon_4.$ Combined with Main Theorem~\ref{mt:2}, Lemma~\ref{pro:3vs4} implies that $EAGSh_{n,R_2,T}$ is $(\theta, \mu, \epsilon)$-LL Resilient for $\theta<T,\mu<\log q$ and $\epsilon=\min(\epsilon_3,\epsilon_4)=\epsilon_3.$

\begin{prop}\label{pro:3vs4}
Let $\frac{1}{5}\log(q)\leq \mu<\log q.$ Then when $N$ is sufficiently large, $\epsilon_3<\epsilon_4.$
\end{prop}
\begin{proof}
Let $\Delta=\frac{\epsilon_3}{\epsilon_4}.$ Again, we aim to show that $\Delta<1.$ By a simple algebraic manipulation, we have

\begin{eqnarray*}
\Delta\leq \frac{c_2^\prime}{q^{\mathfrak{g}_2}\cdot 2^\mu}\cdot \left(\frac{c_2}{2^{5\mu+1}}\right)^{T-\theta+1}\cdot \left(\frac{q c_2^\prime}{2^{5\mu+1}}\right)^{N-2T}.
\end{eqnarray*}
It is easy to see that as $N$ is sufficiently large, the first two terms are at most $1.$ Using similar argument as before, we have $\frac{qc_2^\prime}{2^{5\mu+1}}\leq \frac{c_2^\prime}{2}.$ This can again be shown to be at most $1.$ This completes the proof that $\Delta<1.$
\end{proof}

Lastly, to compare the $\epsilon$ guarantees for $AGSh_{n,R_1,T}$ and $EAGSh_{n,R_2,T},$ we compare $\epsilon_1$ and $\epsilon_3.$
\begin{prop}\label{pro:1vs3}
Assume that $T<\frac{\sqrt{q}}{q-1}N-1$ and $\theta\geq 2T-\frac{\sqrt{q}}{q-1}N+1.$ If $q\geq 4, \epsilon_1<\epsilon_3.$
\end{prop}
\begin{proof}
Let $\Delta=\frac{\epsilon_1}{\epsilon_3}.$ We aim to show that $\Delta<1.$ By a simple algebraic manipulation, we have
\begin{eqnarray*}
\Delta\leq q^{2T-\frac{\sqrt{q}}{q-1}N-\theta+1}\left(\frac{\sin\left(\frac{\pi}{q}\right)}{\sqrt{q}\sin\left(\frac{\pi}{\sqrt{q}}\right)}\right)^{T-\theta+1}.
\end{eqnarray*}
The assumptions on the values of $\theta, T$ and $N$ ensures that the first term is less than $1.$ Now note that $\sqrt{q}\sin\left(\frac{\pi}{\sqrt{q}}\right)$ is an increasing function on $q$ while $\sin\left(\frac{\pi}{q}\right)$ is a decreasing function on $q.$ Hence $\frac{\sin\left(\frac{\pi}{q}\right)}{\sqrt{q}\sin\left(\frac{\pi}{\sqrt{q}}\right)}\leq \frac{1}{2\sqrt{2}}<1$ which proves that $\Delta<1$ as required.
\end{proof}

Combining Propositions~\ref{pro:1vs2},~\ref{pro:3vs4} and~\ref{pro:1vs3}, we obtain the desired result.

\appendix
\section{Supplementary Material}

\subsection{Proof of Lemma~\ref{lem:psfgen}}\label{app:prooflempsfgen}
By Fourier Inversion Formula, the linearity of expectation and additivity of trace function, we have

\begin{align*}
\mathbb{E}_{\mathbf{x}\leftarrow C}\left[\prod_{i=1}^n f_i(x_i)\right]&=\mathbb{E}_{\mathbf{x}\leftarrow C}\left[\prod_{i=1}^n \sum_{\alpha_i\in \F_{p^{w }}} \hat{f}_i(\alpha_i)\cdot \overline{\chi_{\alpha_i}(x_i)}\right]\\
&= \mathbb{E}_{\mathbf{x}\leftarrow C}\left[\sum_{{\boldsymbol \alpha}\in \F_{p^{w }}^n}\prod_{i=1}^n \hat{f}_i(\alpha_i)\cdot \overline{\chi_{\alpha_i}(x_i)}\right]\\
&= \mathbb{E}_{\mathbf{x}\leftarrow C}\left[\sum_{{\boldsymbol \alpha}\in \F_{p^{w }}^n}\left(\prod_{i=1}^n \hat{f}_i(\alpha_i)\right)\cdot\left(\prod_{i=1}^n \omega_p^{-Tr(\alpha_i\cdot x_i)}\right)\right]\\
&=\sum_{{\boldsymbol \alpha}\in \F_{p^{w }}^n}\left(\prod_{i=1}^n \hat{f}_i(\alpha_i)\right)\cdot \mathbb{E}_{\mathbf{x}\leftarrow C}\left[\omega_p^{-Tr(\langle{\boldsymbol\alpha},\mathbf{x}\rangle)}\right].
\end{align*}

Next we consider $\mathbb{E}_{\mathbf{x}\leftarrow C}\left[\omega_p^{-Tr(\langle{\boldsymbol \alpha},\mathbf{x}\rangle)}\right]$ for various values of ${\boldsymbol\alpha}.$ Note that if ${\boldsymbol\alpha}\in C^\bot,$ we have $\langle {\boldsymbol\alpha},\mathbf{x}\rangle=0$ for any $\mathbf{x}\in C.$ Hence $\mathbb{E}_{\mathbf{x}\leftarrow C}\left[\omega_p^{-Tr(\langle{\boldsymbol\alpha},\mathbf{x}\rangle)}\right]=1$ for any ${\boldsymbol\alpha}\in C^\bot.$ So we focus the remainder of the proof for the case when ${\boldsymbol\alpha}\notin C^\bot.$ Consider $\varphi_{{\boldsymbol\alpha}}:C\rightarrow \F_p$ such that for any $\mathbf{x}\in C, \varphi_{{\boldsymbol\alpha}}(\mathbf{x})=Tr(\langle{\boldsymbol\alpha},\mathbf{x}\rangle).$ It is easy to see that $\varphi_{{\boldsymbol\alpha}}$ is $\F_p$-linear, that is, for any $\mathbf{x},\mathbf{y}\in C$ and $\lambda,\mu\in \F_p, \varphi_{{\boldsymbol\alpha}}(\lambda\mathbf{x}+\mu\mathbf{y})= \lambda\varphi_{{\boldsymbol\alpha}}(\mathbf{x})+\mu\varphi_{{\boldsymbol\alpha}}(\mathbf{y}).$ This shows that for any $z\neq z'\in \F_p$ such that $(\varphi_{{\boldsymbol\alpha}})^{-1}(z)$ and $(\varphi_{{\boldsymbol\alpha}})^{-1}(z')$ are both non-empty, we have $\left|(\varphi_{{\boldsymbol\alpha}})^{-1}(z)\right|=\left|(\varphi_{{\boldsymbol\alpha}})^{-1}(z)\right|.$ Next we prove that $\varphi_{\boldsymbol\alpha}$ is surjective. Since ${\boldsymbol\alpha}\notin C^\bot,$ there exists $\mathbf{x}'\in C$ such that $\langle {\boldsymbol\alpha},\mathbf{x}'\rangle\neq 0\in \F_{p^{w }}.$ Due to the linearity of inner product and $C$ along with the fact that $\F_{p^{w }}$ is a field, for any $y\in \F_{p^{w }},$ we can find an appropriate multiplier $\lambda\in \F_{p^{w }}$ such that $\lambda\mathbf{x}'\in C$ and $\langle {\boldsymbol\alpha},\lambda\mathbf{x}'\rangle=y.$ In particular, since the trace function is a surjective function from $\F_{p^{w }}$ to $\F_p,$ there exists $\mathbf{x}\in C$ such that $\varphi_{{\boldsymbol\alpha}}(\mathbf{x})=1.$ Then for any $z\in \F_p,$ it is easy to see that $z\mathbf{x}\in (\varphi_{{\boldsymbol\alpha}})^{-1}(z).$ This shows that for any $z\in \F_p, (\varphi_{{\boldsymbol\alpha}})^{-1}(z)$ is non-empty and they have the same size for all choices of $z,$ which is $p^{k-1}.$ Hence, if ${\boldsymbol\alpha}\notin C^{\bot},$
\begin{align*}
\mathbb{E}_{\mathbf{x}\leftarrow C}\left[\omega_p^{-Tr(\langle{\boldsymbol\alpha},\mathbf{x}\rangle)}\right]&=\frac{1}{|C|}\sum_{\mathbf{x}\in C} \omega_p^{-Tr(\langle{\boldsymbol\alpha},\mathbf{x}\rangle)}\\
&= \frac{1}{p^k}\cdot p^{k-1} \sum_{i=0}^{p-1} \omega_p^{-i}\\
&= 0
\end{align*}
where the last equality is based on the fact that $\omega_p$ is a root of the polynomial $1+x+\cdots+x^{p-1}.$

So we have
\begin{equation*}
\mathbb{E}_{\mathbf{x}\leftarrow C}\left[\omega_p^{-Tr(\langle{\boldsymbol\alpha},\mathbf{x}\rangle)}\right]=\left\{
\begin{array}{cc}
1,&\mathrm{~if~}{\boldsymbol\alpha}\in C^\bot\\
0,&\mathrm{~otherwise}
\end{array}
\right.
\end{equation*}
which completes our proof.
\subsection{Proof of Lemma~\ref{lem:newSDtoDualsum}}\label{app:prooflemnewsdtodualsum}
Recall that for any $S\subseteq\F_{p^{w }}^n$ and ${\boldsymbol\ell}\in \F_2^{\mu\times n},$
\begin{equation*}
Pr_{\mathbf{x}\leftarrow S}\left[{\boldsymbol\tau}(\mathbf{x})={\boldsymbol\ell}\right]=\mathbb{E}_{\mathbf{x}\leftarrow S}\left[\prod_{j}\mathds{1}_{\ell_j}(x_j)\right].
\end{equation*}
Hence,
\begin{align*}
SD({\boldsymbol\tau}(C),{\boldsymbol\tau}(\mathcal{U}_n))&=\frac{1}{2}\sum_{{\boldsymbol\ell}}\left|Pr_{\mathbf{x}\leftarrow C}\left({\boldsymbol\tau}(\mathbf{x})={\boldsymbol\ell}\right)-Pr_{\mathbf{x}\leftarrow \mathcal{U}_n}\left({\boldsymbol\tau}(\mathbf{x})={\boldsymbol\ell}\right)\right|\\
&=\frac{1}{2}\sum_{{\boldsymbol\ell}}\left|\mathbb{E}_{\mathbf{x}\leftarrow C}\left[\prod_j \mathds{1}_{\ell_j}(x_j)\right]-Pr_{\mathbf{x}\leftarrow \mathcal{U}_n}\left({\boldsymbol\tau}(\mathbf{x})={\boldsymbol\ell}\right)\right|\\
&=\frac{1}{2}\sum_{{\boldsymbol\ell}}\left|\sum_{{\boldsymbol\alpha}\in C^\bot}\prod_j\widehat{\mathds{1}_{\ell_j}}(\alpha_j)-Pr_{\mathbf{x}\leftarrow \mathcal{U}_n}\left({\boldsymbol\tau}(\mathbf{x})={\boldsymbol\ell}\right)\right|
\end{align*}
where the last equality is based on Lemma~\ref{lem:psfgen}.
Now note that when $\mathbf{x}\leftarrow \mathcal{U}_n,$ for each $j=1,\cdots, n, x_i$ is identically and uniformly distributed over $\F_{p^{w }}.$ Hence
\begin{align*}
Pr_{\mathbf{x}\leftarrow \mathcal{U}_n}\left({\boldsymbol\tau}(\mathbf{x})={\boldsymbol\ell}\right)&=\prod_{j=1}^n Pr_{x_j\in \F_{p^{w }}} \left(\tau^{(j)}(x_j)=\ell_j\right)\\
&=\prod_{j=1}^n\frac{\left|\left(\tau^{(j)}\right)^{-1}(\ell_j)\right|}{p^{w }}=\prod_{j=1}^n \widehat{\mathds{1}_{\ell_j}}(\mathbf{0}).
\end{align*}
Since $\mathbf{0}\in C^\bot,$
\begin{equation*}
SD({\boldsymbol\tau}(C),{\boldsymbol\tau}(\mathcal{U}_n))=\frac{1}{2}\sum_{{\boldsymbol\ell}}\left|\sum_{{\boldsymbol\alpha}\in C^\bot\setminus\left\{\mathbf{0}\right\}} \prod_{j=1}^n \widehat{\mathds{1}_{\ell_j}}(\alpha_j)\right|.
\end{equation*}

\subsection{Proof of Lemma~\ref{lem:cmpe}}\label{app:prooflemcmpe}
First, note that $\widehat{\mathds{1}_A}(0)=\mathbb{E}_x[\mathds{1}_A(x)\cdot \omega_p^0]=\frac{|A|}{p^{w }}.$ This proves the case when $\alpha=0.$ Now suppose that $\alpha\neq 0.$ Let $|A_i|=t_i.$ Note that since $\alpha\neq 0, \widehat{\mathds{1}_A}(\alpha) = \mathbb{E}_x[\mathds{1}_A(x)\cdot \omega_p^{Tr(\alpha x)}]=\frac{1}{p^{w }}\cdot \omega_p^{\alpha A}$ where $\alpha A$ has the same size as $A.$ So by Lemma~\ref{lem:newsumpthroot},
\begin{align*}
\sum_{i=1}^{2^\mu} \left|\widehat{\mathds{1}_{A_i}}(\alpha)\right|&=\frac{1}{p^{w }}\sum_{i=1}^{2^\mu}|\omega_p^{\alpha A_i}|\leq \frac{1}{p^{w }}\sum_{i=1}^{2^\mu} \frac{p^{w -1}}{\sin(\pi/p)}\cdot \sin(\pi t_i/p^{w })\\
&= \frac{1}{p\sin(\pi/p)}\cdot \sum_{i=1}^{2^\mu}\sin(\pi t_i/p^{w }).
\end{align*}
Note that since $\sin$ is a concave function between $[0,\pi],$ the sum is maximized if all $t_i=\frac{p^{w }}{2^\mu}.$ Hence
\begin{equation*}
\sum_{i=1}^{2^\mu} \left|\widehat{\mathds{1}_{A_i}}(\alpha)\right|\leq \frac{1}{p\sin(\pi/p)}\cdot \sum_{i=1}^{2^\mu}\sin(\pi t_i/p^{w })\leq \frac{2^\mu \sin (\pi/2^\mu)}{p\sin (\pi/p)}
\end{equation*}

\subsection{Proof of Lemma~\ref{lem:maxcmpe}}\label{app:prooflemmaxcmpe}
Note that for $\alpha\neq 0,$ we have $\widehat{\mathds{1}_{A_i}}(\alpha)=p^{-{w }} \omega_p^{\alpha A_i}$ which is well defined. Let $\alpha_i\in \F_{p^{w }}$ be an element that maximizes $|\widehat{\mathds{1}_{A_i}}(\alpha)|.$  Then $|\widehat{\mathds{1}_{A_i}}(\alpha_i)|=p^{-w }\omega_p^{\alpha_i A_i}=|\widehat{\mathds{1}_{\alpha_i A_i}}(1)|.$ For $i=1,\cdots, 2^\mu,$ let $B_i=\alpha_i A_i.$ Then we have $B_1,\cdots, B_{2^\mu}\subseteq \F_{p^{w }}$ such that $\sum_{i=1}^{2^\mu}|B_i|=p^{w }.$ So applying Lemma~\ref{lem:cmpe}, we have
\[\sum_{i=1}^{2^\mu} \max_{\alpha\neq 0} |\widehat{\mathds{1}_{A_i}}(\alpha)|=\sum_{i=1}^{2^\mu}|\widehat{\mathds{1}_{B_i}}(1)|\leq c_\mu.\]

\subsection{Proof of Lemma~\ref{lem:newSDboundpe}}\label{app:prooflemnewsdboundpe}
For simplicity of notation, let $SD=SD({\boldsymbol \tau}(C),{\boldsymbol\tau}(\mathcal{U}_n)).$ Recall that ${\boldsymbol\alpha}\in C^\bot\setminus\left\{\mathbf{0}\right\}$ if and only if there exists ${\boldsymbol\beta}\in \F_{p^{w }}^{n-k}\setminus\left\{\mathbf{0}\right\}$ such that ${\boldsymbol\alpha}={\boldsymbol\beta} H$ where for each $j, \alpha_j = \langle {\boldsymbol\beta},\mathbf{h}_j\rangle.$ By Lemma~\ref{lem:newSDtoDualsum} and Cauchy-Schwarz inequality, we have

\begin{eqnarray*}
SD&=&\frac{1}{2}\sum_{{\boldsymbol\ell}}\left|\sum_{{\boldsymbol\alpha}\in C^\bot\setminus\left\{\mathbf{0}\right\}} \prod_{j=1}^n \widehat{\mathds{1}_{\ell_j}}(\alpha_j)\right|=\frac{1}{2}\sum_{{\boldsymbol\ell}}\left|\sum_{{\boldsymbol\beta}\in \F_{p^{w }}^{n-k}\setminus\left\{\mathbf{0}\right\}} \prod_{j=1}^n \widehat{\mathds{1}_{\ell_j}}(\langle {\boldsymbol\beta},\mathbf{h}_j\rangle)\right|\\
&=&\frac{1}{2}\sum_{{\boldsymbol\ell}}\left|\sum_{{\boldsymbol\beta}\in \F_{p^{w }}^{n-k}\setminus\left\{\mathbf{0}\right\}} \left(\prod_{j\in I_1} \widehat{\mathds{1}_{\ell_j}}(\langle {\boldsymbol\beta},\mathbf{h}_j\rangle)\right)\cdot\left(\prod_{j\in I_2 \sqcup I_3} \widehat{\mathds{1}_{\ell_j}}(\langle {\boldsymbol\beta},\mathbf{h}_j\rangle)\right)\right|\\
&\leq& \frac{1}{2}\sum_{{\boldsymbol\ell}} \sqrt{\sum_{{\boldsymbol\beta}\in \F_{p^{w }}^{n-k}\setminus\left\{\mathbf{0}\right\}} \prod_{j\in I_1} \left|\widehat{\mathds{1}_{\ell_j}}(\langle {\boldsymbol\beta},\mathbf{h}_j\rangle)\right|^2}\\
&&\cdot \sqrt{\sum_{{\boldsymbol\beta}\in \F_{p^{w }}^{n-k}\setminus\left\{\mathbf{0}\right\}} \prod_{j\in I_2\sqcup I_3}\left| \widehat{\mathds{1}_{\ell_j}}(\langle {\boldsymbol\beta},\mathbf{h}_j\rangle)\right|^2}\\
&\leq& \frac{1}{2}\sum_{{\boldsymbol\ell}} \sqrt{\sum_{{\boldsymbol\beta}\in \F_{p^{w }}^{n-k}\setminus\left\{\mathbf{0}\right\}} \prod_{j\in I_1} \left|\widehat{\mathds{1}_{\ell_j}}(\langle {\boldsymbol\beta},\mathbf{h}_j\rangle)\right|^2}\\
&&\cdot \sqrt{\sum_{{\boldsymbol\beta}\in \F_{p^{w }}^{n-k}\setminus\left\{\mathbf{0}\right\}} \prod_{j\in I_2}\left| \widehat{\mathds{1}_{\ell_j}}(\langle {\boldsymbol\beta},\mathbf{h}_j\rangle)\right|^2\cdot \max_{{\boldsymbol\beta}\in \F_{p^{w }}^{n-k}\setminus\left\{\mathbf{0}\right\}}\prod_{j\in I_3}\left|\widehat{\mathds{1}_{\ell_j}}(\langle{\boldsymbol\beta},\mathbf{h}_j\rangle)\right|^2}
\end{eqnarray*}

Recall that since $C^\bot$ has minimum distance $d^\bot\leq k+1,$ any $n-d^\bot+1$ columns of $H$ has full rank. So in particular, since $|I_1|=|I_2|=n-d^\bot+1,$ for any $x=1,2,$ the function ${\boldsymbol\beta}\in \F_{p^{w }}^{n-k}     \mapsto\{\langle{\boldsymbol\beta},\mathbf{h}_j\rangle\}_{j\in I_x}\in \F_{p^{w }}^{n-d^\bot+1}$ is injective. So for any $x=1,2,$
\begin{eqnarray*}
\sum_{{\boldsymbol\beta}\in \F_{p^{w }}^{n-k}\setminus\left\{\mathbf{0}\right\}} \prod_{j\in I_x}\left|\widehat{\mathds{1}_{\ell_j}}\left(\left\langle{\boldsymbol\beta},\mathbf{h}_j\right\rangle\right)\right|^2&\leq& \sum_{\{{\boldsymbol\alpha}_j\}_{j\in I_x}\in \F_{p^{w }}^{n-d^\bot+1}}\prod_{j\in I_x}\left|\widehat{\mathds{1}_{\ell_j}}(\alpha_j)\right|^2\\
&=&\prod_{j\in I_x}\sum_{\alpha\in \F_p}\left|\widehat{\mathds{1}_{\ell_j}}(\alpha)\right|^2=\prod_{j\in I_x} \left\|\widehat{\mathds{1}_{\ell_j}}\right\|_2^2
\end{eqnarray*}
Applying this to the previous inequality, we have
\begin{align*}
SD&\leq \frac{1}{2}\sum_{{\boldsymbol\ell}}\left\|\prod_{j\in I_1}\widehat{\mathds{1}_{\ell_j}}\right\|_2 \cdot\left\|\prod_{j\in I_1}\widehat{\mathds{1}_{\ell_j}}\right\|_2\cdot\max_{{\boldsymbol\beta}\in \F_{p^{w }}^{n-k}\setminus\left\{\mathbf{0}\right\}}\prod_{j\in I_3}\left|\widehat{\mathds{1}_{\ell_j}}(\langle{\boldsymbol\beta},\mathbf{h}_j\rangle)\right|\\
&=\frac{1}{2}\cdot \left(\prod_{j\in I_1\cup I_2}\sum_{\ell_j}\left\|\widehat{\mathds{1}_{\ell_j}}\right\|_2\right)\cdot \sum_{\{\ell_j\}_{j\in I_3}}\max_{{\boldsymbol\beta}\in \F_{p^{w }}^{n-k}\setminus\left\{\mathbf{0}\right\}}\prod_{j\in I_3}\left|\widehat{\mathds{1}_{\ell_j}}(\langle{\boldsymbol\beta},\mathbf{h}_j\rangle)\right|
\end{align*}
We conclude the proof by proving that $\prod_{j\in I_1\sqcup I_2}\sum_{\ell_j}\left\|\widehat{\mathds{1}_{\ell_j}}\right\|_2\leq 2^{\mu(n-d^\bot+1)}.$ Note that to prove this, it is sufficient to prove the following proposition.

\begin{prop} \label{prop:fourier2normbound}
For any $j\in [n]$ and $\mu,$ we have $\sum_{\ell_j\in \F_2^\mu}\left\|\widehat{\mathds{1}_{\ell_j}}\right\|_2\leq 2^{\mu/2}.$
\end{prop}
\begin{proof}
Recall that by Parseval's Identity in Lemma~\ref{lem:PIFIF}, we have $\|\widehat{\mathds{1}_{\ell_j}}\|_2 = \|\mathds{1}_{\ell_j}\|_2= \sqrt{\mathbb{E}_{x\leftarrow \F_{p^{w }}}[|\mathds{1}_{\ell_j}(x)|^2]}$ $=\sqrt{Pr_{x\in \F_{p^{w }}}[\mathds{1}_{\ell_j}(x)=1]}.$ For any $\ell\in \F_2^\mu,$ denote by $S_{\ell}=\{x\in \F_{p^{w }}:\mathds{1}_\ell(x)=1\}.$ It is easy to see that $\F_{p^{w }}=\bigsqcup_{\ell\in \F_2^\mu}S_\ell.$ Hence we have $\|\widehat{\mathds{1}_{\ell_j}}\|_2=\sqrt{\frac{|S_{\ell_j}|}{p^{w }}}, \sum_{\ell_j\in \F_2^\mu}\|\widehat{\mathds{1}_{\ell_j}}\|_2=\sum_{\ell_j\in \F_2^\mu}\sqrt{\frac{|S_{\ell_j}|}{p^{w }}}$ and $\sum_{\ell_j\in \F_2^\mu}|S_{\ell_j}|=p^{w }.$

Recall that for any non-negative real numbers $a_1,\cdots, a_{2^\mu},$ the relation between their arithmetic mean and quadratic mean tell us that
$\sum_{i=1}^{2^\mu} a_i$ is upper bounded by $\sqrt{2^\mu} \sqrt{\sum_{i=1}^{2^\mu} a_i^2}.$ So by setting the non-negative real numbers to be $\|\widehat{\mathds{1}_{\ell_j}}\|_2,$ we have
\[\sum_{\ell_j\in \F_2^\mu}\|\widehat{\mathds{1}_{\ell_j}}\|_2=\sum_{\ell_j\in \F_2^\mu}\sqrt{\frac{|S_{\ell_j}|}{p^{w }}}\leq 2^{\mu/2}\sqrt{\sum_{\ell_j\in \F_2^\mu} \frac{|S_{\ell_j}|}{p^{w }}}=2^{\mu/2}.\]
\end{proof}
Applying Proposition~\ref{prop:fourier2normbound}, we have the desired bound for $SD({\boldsymbol \tau}(C),{\boldsymbol \tau}(\mathcal{U}_n)).$

\subsection{Proof of Lemma~\ref{lem:newcmgenpe}}\label{app:prooflemnewcmgenpe}
Denote by $\eta=\sum_{{\boldsymbol\ell}\in \F_2^{\mu\times \kappa}} \max_{{\boldsymbol\alpha}\in D\setminus\{\mathbf{0}\}} \prod_{j=1}^{\kappa} \left|\widehat{\mathds{1}_{\ell_j}}(\alpha_j)\right|.$ Here we prove a bound for $\sum_{{\boldsymbol\ell}\in \F_2^{\mu\times \kappa}} \left|\widehat{\mathds{1}_{A_i}}(\alpha)\right|$ that we will use instead of the bounds  from Lemma~\ref{lem:cmpe}.
\begin{lem}\label{lem:newxi}
Let $\zeta_{p^{w }}:\{0,\cdots, p^{w }-1\}\rightarrow \mathbb{R}_{\geq 0}$ such that $\zeta_{p^{w }}(x) = \frac{p^{w } \cdot \sin(\pi x/p^{w })}{p\sin (\pi/p)}.$ We further let $\xi_{p^{w }}:\{0,\cdots,p^{w -1}\}\rightarrow \mathbb{R}_{\geq 0}$ be defined such that for $\xi_{p^{w }}(x)\triangleq \max(\zeta_{p^{w }}(x)/p^{w }, 2^{-(4\mu+1)}).$ Then the function $\xi_{p^{w }}$ has the following properties:
\begin{enumerate}
\item For every set $A$ of size $t,$
\begin{equation*}
\left|\widehat{\mathds{1}_A}(\alpha)\right|\leq\left\{
\begin{array}{cc}
\xi_{p^{w }}(t),&\mathrm{~if~} \alpha\neq 0\\
2^{4\mu+1}\cdot \xi_{p^{w }}(t), &\mathrm{~if~} \alpha=0.
\end{array}
\right.
\end{equation*}
\item Let $A_1,\cdots, A_{2^\mu}$ be any partition of $\F_{p^{w }}.$ Then,
\begin{equation*}
\sum_{i=1}^{2^\mu} \xi_{p^{w }}(|A_i|)\leq c_\mu^\prime.
\end{equation*}
\end{enumerate}
\end{lem}
First we prove Lemma~\ref{lem:newxi} before using it to complete the proof of Lemma~\ref{lem:newcmgenpe}.
\begin{proof}[Proof of Lemma~\ref{lem:newxi}]
\begin{enumerate}
\item Note that since $\xi_{p^{w }}(t)\geq \zeta_{p^{w }}(t),$ the inequality when $\alpha\neq 0$ follows from Lemma~\ref{lem:newsumpthroot}. Noting that $\left|\widehat{\mathds{1}_A}(0)\right|=\frac{|A|}{p^{w }}$ and $2^{4\mu+1}\cdot \xi_{p^{w }}(x)\geq 2^{4\mu+1}\cdot 2^{-(4\mu+1)}=1,$ the inequality for the case $\alpha=0$ directly follows.
\item First we note that $\zeta_{p^{w }}(p^{w }/2^{4\mu})/{p^{w }} = \frac{\sin (\pi/2^{4\mu})}{p\sin(\pi/p)}.$ Note that $x-\sin(x)$ is a non-decreasing function and it is zero only if $x=0.$ So $p\sin(\pi/p)\leq p\cdot \frac{\pi}{p}=\pi\leq 4.$ We further note that $\frac{\sin x}{x}$ is a decreasing function for $x\in(0,\pi/2]$ by the same reason as above. This means that for any $x<y, \frac{\sin x}{x} \geq \frac{\sin y}{y}$ or equivalently, $\frac{\sin x}{\sin y}\geq \frac{x}{y}.$ So setting $x=\frac{\pi}{2^{4\mu}}$ and $y=\frac{\pi}{2},$ we have $\sin(\pi/2^{4\mu})\geq \frac{2}{\pi}\cdot \frac{\pi}{2^{4\mu}}= 2^{-4\mu+1}.$ So $\frac{\zeta_{p^{w }}(p^{w }/2^{4\mu})}{p^{w }} \geq 2^{-(4\mu+1)}.$ Hence we have $\xi_{p^{w }}(x)=\max\left(\frac{\zeta_{p^{w }}(x)}{p^{w }},2^{-(4\mu+1)}\right)\leq\max\left(\frac{\zeta_{p^{w }}(x)}{p^{w }},\frac{\zeta_{p^{w }}(p^{w }/2^{4\mu})}{p^{w }}\right)$ and noting that $\zeta_{p^{w }}(x)$ is an increasing function for $x\in[0,\pi/2],$ we have $\xi_{p^{w }}(x)\leq \frac{1}{p^{w }} \zeta_{p^{w }}\left(\max\left(x,\frac{p^{w }}{2^{4\mu}}\right)\right).$ Now suppose that $|A_i|=t_i$ where $\sum_{i=1}^{2^\mu} t_i = p^{w },$ we have
\begin{align*}
\sum_{i=1}^{2^\mu} \xi_{p^{w }}(t_i)&\leq \frac{1}{p^{w }}\sum_{i=1}^{2^\mu}\zeta_{p^{w }}\left(\max\left(t_i,\frac{p^{w }}{2^{4\mu}}\right)\right)\\
&=  \frac{1}{p \sin(\pi/p)}\sum_{i=1}^{2^\mu}\sin\left(\frac{\pi}{p^{w }}\cdot\max\left(t_i,\frac{p^{w }}{2^{4\mu}}\right)\right)
\end{align*}
We note that due to the concavity of the sine function in the range $[0,\pi],$ for any $a_1,\cdots, a_{2^\mu}\in [0,\pi],\sum \sin(a_i)\leq 2^\mu \sin\left(\frac{\sum a_i}{2^\mu}\right).$ So we have
\begin{align*}
\sum_{i=1}^{2^\mu} \xi_{p^{w }}(t_i)&\leq \frac{2^\mu}{p \sin(\pi/p)}\sin\left(\frac{1}{2^\mu}\cdot \frac{\pi}{p^{w }}\cdot\left(\sum_{i=1}^{2^\mu}\max\left(t_i,\frac{p^{w }}{2^{4\mu}}\right)\right)\right)\\
&=\frac{2^\mu}{p\sin(\pi/p)}\sin\left(\frac{\pi}{2^\mu}\left(\max\left(1,\frac{2^\mu}{2^{4\mu}}\right)\right)\right)\\
&\leq\frac{2^\mu}{p\sin(\pi/p)}\sin\left(\frac{\pi}{2^\mu}\left(1+\frac{2^\mu}{2^{4\mu}}\right)\right)\\
&=\frac{2^\mu \sin(\pi/2^\mu+\pi/2^{4\mu})}{p\sin \pi/p}=c_\mu^\prime.
\end{align*}
\end{enumerate}
This completes the proof for Lemma~\ref{lem:newxi}.
\end{proof}
Now we continue the proof of Lemma~\ref{lem:newcmgenpe}. For any $j$ and $\ell_j\in \F_2^\mu,$ we denote $t_{\ell_j,j}=|\tau_j^{-1}(\ell_j)|.$ By the first claim in Lemma~\ref{lem:newxi}, noting that $wt_H({\boldsymbol\alpha})\geq d$ for any ${\boldsymbol\alpha}\in D\setminus\{\mathbf{0}\},$
\begin{align*}
\eta&\leq \sum_{{\boldsymbol\ell}\in \F_2^{\mu\times \kappa}} \max_{{\boldsymbol\alpha}\in D\setminus\{\mathbf{0}\}}\prod_{j=1}^{\kappa} \xi_{p^{w }}(t_{\ell_j,j})\cdot(2^{4\mu+1})^{\mathds{1}_0(\alpha_j)}\\
&= \sum_{{\boldsymbol\ell}\in \F_2^{\mu\times \kappa}} \prod_{j=1}^{\kappa} \xi_{p^{w }}(t_{\ell_j,j})\cdot \max_{{\boldsymbol\alpha}\in D\setminus\{\mathbf{0}\}}\prod_{j=1}^\kappa 2^{(4\mu+1)(\mathds{1}_0(\alpha_j)}\\
&=\sum_{{\boldsymbol\ell}\in \F_2^{\mu\times \kappa}} \prod_{j=1}^{\kappa} \xi_{p^{w }}(t_{\ell_j,j})\cdot \max_{{\boldsymbol\alpha}\in D\setminus\{\mathbf{0}\}} 2^{(4\mu+1)(\kappa-wt_H({\boldsymbol\alpha}))}\\
&\leq\sum_{{\boldsymbol\ell}\in \F_2^{\mu\times \kappa}} \prod_{j=1}^{\kappa} \xi_{p^{w }}(t_{\ell_j,j})\cdot \max_{{\boldsymbol\alpha}\in D\setminus\{\mathbf{0}\}} 2^{(4\mu+1)(\kappa-d)}\\
&= 2^{(4\mu+1)(\kappa-d)}\cdot \prod_{j=1}^\kappa \sum_{\ell_j\in\F_2^m}\xi_{p^{w }}(t_{\ell_j,j}).
\end{align*}
Note that for any $j, \sum_{\ell_j\in \F_{2^m}} t_{\ell_j,j}=p^{w }.$ So using $\{\tau_j^{-1}(\ell_j)\}_{\ell_j\in \F_2^m}$ as a partition of $\F_{p^{w }},$ by the second claim of Lemma~\ref{lem:newxi}, we have
\begin{equation*}
\prod_{j=1}^\kappa \sum_{\ell_j\in\F_2^m}\xi_{p^{w }}(t_{\ell_j,j})\leq \prod_{j=1}^\kappa c_\mu^\prime= (c_\mu^\prime)^\kappa.
\end{equation*}
Combining this inequality with the upper bound of $\eta$ above, we have the desired inequality.

\end{document}